\documentclass[a4paper,UKenglish,anonymous=true,cleveref, autoref, thm-restate]{lipics-v2021}

\usepackage{env}

\title{A Complete Graphical Language for Linear Optical Circuits with Finite-Photon-Number Sources and~Detectors}

\author{Nicolas Heurtel}{Quandela, 7 Rue Léonard de Vinci, 91300 Massy, France \\Université Paris-Saclay, CNRS, ENS
Paris-Saclay, Inria, Laboratoire Méthodes Formelles, 91190,
Gif-sur-Yvette, France}{nicolas.heurtel@quandela.com}{https://orcid.org/0000-0002-9380-8396}{}

\authorrunning{N. Heurtel}
\titlerunning{The \love-Calculus}



\category{} 

\relatedversion{} 

\supplement{}

\funding{This work has been partially funded by the European Commission as part of the EIC accelerator program, under the grant agreement 190188855, by the French National Research Agency (ANR) with the project TaQC ANR-22-CE47-0012 and within the framework of ``Plan France 2030'', under the research projects EPIQ ANR-22-PETQ-0007, OQULUS ANR-23-PETQ-0013, HQI-Acquisition ANR-22-PNCQ-0001 and HQI-R\&D ANR-22-PNCQ-0002, and by the CIFRE 2022/0081.}

\nolinenumbers
\hideLIPIcs

\begin{document}

\maketitle

\begin{abstract}
    Graphical languages are powerful and useful to represent, rewrite and simplify different kinds of processes. In particular, they have been widely used for quantum processes, improving the state of the art for compilation, simulation and verification. In this work, we focus on one of the main carrier of quantum information and computation: linear optical circuits. 
    We introduce the \love-calculus, the first graphical language to reason on the infinite-dimensional photonic space with circuits only composed of the four core elements of linear optics: the phase shifter, the beam splitter, and auxiliary sources and detectors with bounded photon number. First, we study the subfragment of circuits composed of phase shifters and beam splitters, for which we provide the first minimal equational theory. Next, we introduce a rewriting procedure on those \love-circuits that converge to normal forms. We prove those forms to be unique, establishing both a novel and unique representation of linear optical processes. Finally, we complement the language with an equational theory that we prove to be complete: two \love-circuits represent the same quantum process if and only if one can be transformed into the other with the rules of the \love-calculus.
\end{abstract}

\tableofcontents
\markboth{N. Heurtel}{The \love-Calculus}


\section{Introduction}

\label{sec:intro}

Quantum computing is a paradigm for processing information~\cite{nielsen2010quantum,preskill2023quantum} that performs computation with quantum states, instead of the classical states of bits. This computational paradigm allows specific computational problems to be solved with quadratic~\cite{grover1996algorithm} or even exponential speedup~\cite{shor1994algorithm,harrow2009hhl} compared to their classical counterparts. To encode that quantum data, many technologies have been pursued, such as superconducting circuits~\cite{kjaergaard2020superconducting}, trapped ions~\cite{bruzewicz2019trapped-ion} and cold atoms~\cite{gross2017atoms}.

One of the prominent candidates for quantum computation is linear optics~\cite{kok2007review,oBrien2009photonic,rudolph2017why}, where the \textit{logical} information is encoded into the quantum states of photons, the \textit{particles} of light. For quantum computation, the logical states are encoded onto the \textit{modes} of the photons, i.e.\ their degrees of freedom like their \textit{positions} in the circuit, and the desired logical operations are performed with optical components. All scalable quantum computations with linear optics \cite{knill2001klm,yoran2003deterministic,nielsen2004optical,browne2005resource,bartolucci2023fbqc,degliniasty2024spoqc} encoding with the positions of the photons use predominantly these following elements.
\begin{itemize}
    \item Sources: they generate the quantum state, i.e.\ a vector in a Hilbert space,
    \item Phase shifters: they change the quantum state by adding a phase to the light passing through them\footnote{They are typically implemented using thermo-optic components, where the refractive index of the waveguide is changed by heating the material.},
    \item Beam splitters: they alter the quantum state by causing photons on two different paths to interfere with each other\footnote{ In integrated circuits, beam splitters are implemented using waveguides that split and combine light paths.},
    \item Detectors: they project the quantum state on a subspace. 
\end{itemize}

As ubiquitous as the circuits made of those components are in linear optical quantum computation schemes, as illustrated in \pref{fig:ex-CNot}~and~\ref{fig:ex-Bell}, many unanswered questions persist regarding optimality, minimality and an efficient use of those components. 
We wish to have a framework finding the most appropriate implementation for the desired computation or protocol. The purpose of this work is therefore to propose a formal framework to model and manipulate generic circuits composed of the four previous elements.

\textbf{State of the art.} Some main formal frameworks to study, develop or optimize quantum processes are graphical languages~\cite{ambramsky2004categorical,selinger2007dagger,abramsky2008categorical,penrose1971angular}, representing processes with diagrams and equations between those diagrams. These formalisms have been shown to be very useful for addressing quantum processes in general, such as \ZX-diagrams~\cite{coecke2011interacting} with applications in compilation~\cite{kissinger2019tcount,backens2021therebackagain,vandeWetering2024optimal}, simulation~\cite{kissinger2022classical,kissinger2022simulating,koch2023speedy} and verification~\cite{duncan2014verifying,garvie2018verifying}. To completely capture the processes those diagrams model,~\cite{jeandel2018complete,hadzihasanovic2018complete} have introduced a complete set of equations: two equivalent \ZX-diagrams can always be transformed from one to the other with those equations.

Recently, some works have modeled optical processes with diagrams~\cite{ataman2018graph,clement2020pbs,mccloud2022category}, including notably \lovbf~\cite{clement2022lov}, a complete graphical language for linear optical circuits with vacuum sources and detectors, and \qpath~\cite{felice2022path}, a graphical language to compute amplitudes. Remarkably, both have also led to results beyond the optical realm, as a subfragment of the first led to derive the first complete equational theory for quantum circuits~\cite{clement2023complete}, while the second introduced a functor from the \ZX-calculus~\cite{felice2022path} and led to a more generic language~\cite{deFelice2023light-matter}.

However, those two frameworks don't completely capture linear optical circuits with sources and detection schemes. In particular, \lovbf~lacks a many-photon semantics and can only cover the single-photon case, while \qpath~uses sums of diagrams in the rewriting process along with generators that are not linear optical components. For instance, we would like to be able to model the photonic implementation of the CZ gate~\cite{knill2001klm,knill2002note}, a prominent logical quantum gate, and rewrite it to equivalent forms, as illustrated in \pref{fig:ex-CNot}. 


\begin{figure}[ht]
    \[\tikzfigbox{0.75}{exCZconv2}=\tikzfigbox{0.75}{exCZlofi}\]
    \caption{\label{fig:ex-CNot}Optical circuits implementing the CZ two-qubit logic gate with auxiliary sources and detectors. On the left is the original circuit\footnotemark~of~\cite{knill2002note}. There are two auxiliary photon generated on the bottom left: if exactly one photon is detected for each of the two last wires on the bottom right, then we know we have performed the operation $\ket{11}\mapsto -\ket{11}$ on the two first wires. This event has a probability $\frac{2}{27}$ to occur. On the right is an equivalent representation in the \love-calculus, where $\vect{f}$ and $\vect{g}$ are two-photon states and linear forms.}
\end{figure}
\footnotetext{Some phases have been added to take into account the different conventions for the semantics of the beam splitters.}
\begin{figure}[ht]
    \[\tikzfigbox{0.90}{exBellconv2}=\input{exBellgen.tikz}\]
    \caption{\label{fig:ex-Bell}Linear optical circuit generating with a $\frac{1}{9}$ probability the Bell state $\ket{\Phi^+}=\ket{1010}+\ket{0101}$, with the use of auxiliary sources and detectors. On the left is the original\footnotemark[\value{footnote}]~circuit of~\cite{fldzhyan2021compact}, on the right is an equivalent and modular description. Both are equivalent circuits in the \love-calculus.}
\end{figure}

\textbf{Challenges.} In seeking to develop a graphical language for modeling linear optical circuits with a many-photon semantics, there are two main challenges. First, the bosonic Fock space, representing all the states that photons can be in, is an infinite-dimensional Hilbert space: the bosonic Fock space. In particular, some properties and generators of graphical languages with finite-dimensional theories~\cite{poor2023completeness,wang2024completeness,deVisme2024minimality} cannot be used. Second, the interaction of photons, even without bringing auxiliary modes and detections into the picture, is described by the permanents of matrices~\cite{scheel2004permanents,aaronson2011computational}, making cumbersome explicit expression and manipulation of all the involved terms.
    
\textbf{Contributions.} In this paper, we propose such a framework, and introduce the \love-calculus, the first graphical language defined on the bosonic Fock space, with circuits composed of four core elements of linear optics: the phase shifter, the beam splitter, and auxiliary finite-photon-number sources and detectors. Our main contributions are the following.
\begin{itemize}
    \item A complete equational theory for circuits with phase shifters and beam splitters which is simpler than the one in~\cite{clement2022lov}, and that we prove to be minimal (\pref{sec:loqc}).
    \item A new sound and complete equational theory for linear optical circuits that allows all auxiliary finite-photon-number states and detections~(\pref{sec:love}).
    \item A unique and compact universal form for optical circuits of this kind, obtained through a deterministic rewriting procedure and proven to be unique with new techniques (\pref{sec:completeness}).
\end{itemize}
All the notation introduced in the paper is summarized in \pref{tab:notations}.


\section{\lopp: Linear optical quantum circuits with single-photon semantics}

\label{sec:loqc}

A linear optical quantum circuit consists of spatial modes through which photons pass --- which may be physically instantiated by optical fibers, waveguides in integrated circuits, or simply by paths in free space (bulk optics) --- and operations that act on those spatial modes, including in particular \textit{beam splitters} ($\tikzfigbox{\bx}{bs-xs}$), and \textit{phase shifters} ($\tikzfigbox{\bx}{ps-xs}$).

\subsection{Syntax and single-photon semantics}
\label{subsec:singlephoton}

Similarly to~\cite{clement2022lov}, in order to formalize linear optical circuits with diagrams, we use the formalism of PROPs~\cite{maclane1965category}.
A PRO is a strict monoidal category whose monoid of objects is freely generated by a single $X$: the objects are all of the form $X\otimes X \otimes \dots \otimes X$, and simply denoted by $n$, the number of
occurrences of $X$. PROs are typically represented graphically as circuits: each copy of $X$ is represented by a wire and morphisms by boxes on wires, so that $\otimes$ is represented vertically and morphism composition $\circ$ is represented horizontally. For instance, $D_1$ and $D_2$ represented as \tikzfigbox{\bx}{D1} and \tikzfigbox{\bx}{D2} can be horizontally composed as $D_2\circ D_1$, represented by \tikzfigbox{\bx}{D1compD2}, and vertically composed as $D_1\otimes D_2$, represented by \tikzfigbox{\bx}{D1tensD2}. A PROP is the symmetric monoidal analogue of PRO, so it is equipped with a swap. It means the circuits are equivalent through wire deformations and that only the connectivity matters.

\begin{definition}\lopp\footnote{The PROP version of \lopp~has first been defined in~\cite{clement2023minimal}, as~\cite{clement2022lov} defined $\textbf{\textup{LO}}_\textbf{\textup{PP}}$ as a PRO.} is the PROP of \lopp-circuits generated by:

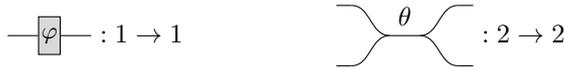
\[\begin{tikzpicture}
	\begin{pgfonlayer}{nodelayer}
		\node [style=none] (5) at (-1.1, 0) {};
		\node [style=none] (9) at (1.1, 0) {};
		\node [style=phase] (11) at (0, 0) {\small $\varphi$};
	\end{pgfonlayer}
	\begin{pgfonlayer}{edgelayer}
		\draw (5.center) to (11);
		\draw (11) to (9.center);
	\end{pgfonlayer}
\end{tikzpicture}
:1\rightarrow 1 \hspace{2cm}\input{bs.tikz}:2\rightarrow 2 \]

with $\varphi\in\mathbb{R}$ and $\theta\in\mathbb{R}$.
\end{definition}

The convention is to go through from left to right, meaning the inputs (resp.\ outputs) are on the left (resp.\ right), and from top to bottom, meaning the first (resp.\ last) input is the top (resp.\ last) wire. The identity, the swap and the empty diagrams are noted as follows: $\tikzfigbox{\bxs}{id},\tikzfigbox{\bxs}{swap},\tikzfigbox{\bxs}{empty}$. 

\begin{example} Here are two examples of \lopp-circuits, that are equivalent up to deformation with the rules of PROPs: \[\tikzfigbox{\bx}{exLOppleft}=\tikzfigbox{\bx}{exLOppright}\]
\end{example}

The semantics of linear optical components are usually described by their behavior when there is one single photon passing through those components. 
Given a circuit of $m$ wires, the single photon can be in a superposition of the $m$ different positions, so its state is a vector in $\C^m$. We consider the standard orthonormal basis $\{ \ket{e_i},i\in[1,m]\}$ where $e_i=\ket{0,\dots,0,1,0,\dots,0}$ with a $1$ at the $i^{th}$ component. The object of our PROP is therefore $X=\C$, and the vertical composition is interpreted as the direct sum~\cite{clement2022lov,clement2023complete}. The semantics is defined as follows. 

\begin{definition}[Semantics of \lopp] \label{def:semLOpp}

The single photon semantics of \lopp~is inductively defined as follows: 
$\interpsone{C_1 \otimes C_2} = \interpsone{C_1} \oplus \interpsone{C_2}$, $\interpsone{C_2 \circ C_1}=\interpsone{C_2}\circ \interpsone{C_1}$ and:

$\interpsone{\tikzfigbox{\bx}{id}}:\C \to \C \coloneqq \ket{1}\mapsto \ket{1}$

$\interpsone{\begin{tikzpicture}[scale=0.85]
	\begin{pgfonlayer}{nodelayer}
		\node [style=none] (5) at (-1.1, 0) {};
		\node [style=none] (9) at (1.1, 0) {};
		\node [style=phase] (11) at (0, 0) {\small $\varphi$};
	\end{pgfonlayer}
	\begin{pgfonlayer}{edgelayer}
		\draw (5.center) to (11);
		\draw (11) to (9.center);
	\end{pgfonlayer}
\end{tikzpicture}
}:\C \to \C \coloneqq \ket{1}\mapsto e^{i\varphi}\ket{1}$

$\interpsone{\tikzfigbox{\bx}{swap}}:\C^2 \to \C^2 \coloneqq \begin{array}{ccc}
    \ket{1,0} & \mapsto & \ket{0,1}  \\
     \ket{0,1} & \mapsto & \ket{1,0}
\end{array}$

$\interpsone{\input{bs-xs.tikz}}:\C^2 \to \C^2 \coloneqq \begin{array}{ccc}
    \ket{1,0} & \mapsto & \costheta\ket{1,0} + i\sintheta\ket{0,1}  \\
     \ket{0,1} & \mapsto & i\sintheta\ket{1,0}+\costheta\ket{0,1}
\end{array}$

where $\costheta=\cos(\theta)$ and $\sintheta=\sin(\theta)$.
\end{definition}

\begin{remark}
It is also usual to represent those linear operators as matrices,
with\\ $\interpsone{C_1}\!\oplus \interpsone{C_2} = \left(\begin{array}{c|c} \interpsone{C_1} & 0\\ \hline 0 & \interpsone{C_2} \end{array}\right)$ and for instance $\interpsone{\input{bs-xs.tikz}}=\begin{pmatrix}
\costheta & i\sintheta \\
i\sintheta &  \costheta 
\end{pmatrix}$. 
\end{remark}


\subsection{Simpler equational theory of~\lopp}\label{subsec:eqlopp}

Two distinct \lopp-circuits may represent the same quantum evolution. For instance, shifting the phase of a photon by two phase shifters of phase $\varphi_1$ and $\varphi_2$ is the same as shifting it with one phase $\varphi_1+\varphi_2$. In order to characterize those equivalences, an equational theory of \lopp~has been introduced in~\cite{clement2022lov}. In this section, we provide a simpler set of equations in \pref{fig:nLOpp}. Some of the old equations, given in~\pref{fig:oLOpp}, have been removed, while two Equations~(\ref{oLOpp:E2}) and (\ref{oLOpp:E3}) of \pref{fig:oLOpp} have been replaced by the two Equations~(\ref{nLOpp:E2}) and (\ref{nLOpp:E3}), respectively representing Euler rotations with two and three axes. Previously, those old Euler equations were not directly reversible; while the angles of the right-hand side (RHS) could be uniquely determined by those of the left-hand side (LHS), the inverse was true only with non-trivial constraints, making the equations hardly reversible and not explicitly constructive. More specifically, we made the following simplifications:

\begin{itemize}
    \item The Equations~(\ref{nLOpp:b0}), (\ref{nLOpp:p0}) and~(\ref{oLOpp:pp}) have been derived and removed from the equational theory.
    \item A phase has been added in Equation~(\ref{oLOpp:E2}), so now the LHS can also represent any element of the unitary group $U(2)$. Now the angles of the LHS can be straightforwardly derived without any constraints from the RHS. 
    \item All the phases of Equation~(\ref{oLOpp:E3}) have been removed. The formulae of the equations are now far simpler, and the equation is now both symmetrical and reversible. 
\end{itemize}

\newcommand{\sizenlopp}{0.85}
\begin{figure}[htb]
    \centering
    \fbox{\begin{minipage}{\textwidth}\begin{multicols}{2}
    \renewcommand{\nspazer}{0em}
    \scalebox{0.9}{\begin{minipage}{0.50\textwidth}
        \begin{equation}\tag{p2$\pi$}\label{nLOpp:Rp2pi}\begin{array}{rccl}\tikzfigbox{\sizenlopp}{LOpp-A-middle}&\!\!\!\!\!=\!\!\!\!\!&\tikzfigbox{\sizenlopp}{LOpp-A-right}\end{array}\end{equation}
    \begin{equation}\tag{swap} \label{nLOpp:swap}\begin{array}{rcl}\tikzfigbox{\sizenlopp}{swap}&\!\!\!\!\!=\!\!\!\!\!&\tikzfigbox{\sizenlopp}{swapright}\end{array}\end{equation}\end{minipage}}\scalebox{0.9}{\begin{minipage}{0.6\textwidth}
    \begin{equation}\tag{p-p} \label{nLOpp:pM}\begin{array}{rcl}\tikzfigbox{\sizenlopp}{LOpp-D-left}&\!\!\!\!\!=\!\!\!\!\!&\tikzfigbox{\sizenlopp}{LOpp-D-right}\end{array}\end{equation}
    \begin{equation}\tag{E2} \label{nLOpp:E2}\begin{array}{rcl}\tikzfigbox{\sizenlopp}{E2left}&\!\!\!\!\!=\!\!\!\!\!&\tikzfigbox{\sizenlopp}{E2right}\end{array}\end{equation}\end{minipage}}
    \end{multicols}
    \vspace{0.1cm}
    \scalebox{0.9}{\begin{minipage}{1.1\textwidth}\begin{equation}\tag{E3}\label{nLOpp:E3}\begin{array}{rcl}~\qquad\tikzfigbox{\sizenlopp}{E3left}&=&\tikzfigbox{\sizenlopp}{E3right}\end{array}\end{equation}\end{minipage}}\vspace{0.2cm}
\end{minipage}}
    \caption{\label{fig:nLOpp} New and minimal equational theory of the \lopp-calculus.  For any angle of the LHS (resp. RHS) of the Equation~(\ref{nLOpp:E2}) and (\ref{nLOpp:E3}), there exist angles for the RHS (resp. LHS) such that the equations are sound. The angles of are unique if we restrict $\alpha_0,\alpha_2,\beta_0,\beta_1,\beta_3\in[0,2\pi)$, $\alpha_1 \in [0,\frac{\pi}{2}),\alpha_3\in[0,\pi)$, $\beta_2\in[0,\frac{\pi}{2}]$, and by taking $\alpha_1=0$ if $\alpha_0-\alpha_2=0~\text{mod}~\pi$ and $\beta_1=0$ if $\beta_2\in\left\{0,\frac{\pi}{2}\right\}$. The rotations associated with Equations~(\ref{nLOpp:E2}) and~(\ref{nLOpp:E2}) are detailed in the proof of \pref{prop:soundLOpp}, and the explicit values of the angles are detailed in \pref{subsec:E2angles} and~\ref{subsec:E3angles}.}
    \end{figure}

\begin{figure}[htb]
    \centering
    \fbox{\begin{minipage}{\textwidth}
        \begin{multicols}{2}
            \scalebox{0.9}{\begin{minipage}{0.48\textwidth}
                \begin{equation}\tag{p0}\label{nLOpp:p0}\begin{array}{rccl}\tikzfigbox{\sizenlopp}{LOpp-A-left}&=&\tikzfigbox{\sizenlopp}{LOpp-A-right}\end{array}\end{equation}
    \begin{equation}\tag{b0} \label{nLOpp:b0}\begin{array}{rcl}\tikzfigbox{\sizenlopp}{LOpp-B-left}&\!\!\!\!\!=\!\!\!\!\!&\tikzfigbox{\sizenlopp}{LOpp-B-right}\end{array}\end{equation}\end{minipage}}\scalebox{0.9}{\begin{minipage}{0.62\textwidth}
        \begin{equation}\tag{pp-b} \label{oLOpp:pp}\begin{array}{rcl}\tikzfigbox{\sizenlopp}{LOpp-E-left}&\!\!\!\!\!=\!\!\!\!\!&\tikzfigbox{\sizenlopp}{LOpp-E-right}\end{array}\end{equation}
    \begin{equation}\tag{oE2} \label{oLOpp:E2}\begin{array}{rcl}\tikzfigbox{\sizenlopp}{LOpp-F-left}&\!\!\!\!\!=\!\!\!\!\!&\tikzfigbox{\sizenlopp}{LOpp-F-right}\end{array}\end{equation}
\end{minipage}}\end{multicols}
    \scalebox{0.9}{\begin{minipage}{1.1\textwidth}\begin{equation}\tag{oE3}\label{oLOpp:E3}\begin{array}{rcl}~\qquad\qquad\tikzfigbox{\sizenlopp}{LOpp-G-left}&\!\!\!\!\!=\!\!\!\!\!&\tikzfigbox{\sizenlopp}{LOpp-G-right}\end{array}\end{equation}
    \end{minipage}}\vspace{0.1cm}\end{minipage}}
    \caption{\label{fig:oLOpp}Old axioms of the \lopp-calculus that are not in \pref{fig:nLOpp}. They are derived in \pref{app:proofcompLOpp}.  In Equations~(\ref{oLOpp:E2}) and (\ref{oLOpp:E3}), the LHS circuit has arbitrary parameters which uniquely determine the parameters of the RHS circuit. For any $\alpha_i\in\mathbb R$, there exist $\beta_j\in[0,2\pi)$ such that \pref{oLOpp:E2} is sound, and for any $\gamma_i\in\mathbb R$, there exist $\delta_j\in[0,2\pi)$ such that \pref{oLOpp:E3} is sound. We can ensure that the angles $\beta_j$ are unique by assuming that $\beta_1,\beta_2\in [0,\pi)$ and if $\beta_2\in\{0,\frac\pi2\}$ then $\beta_1=0$, and we can ensure that the angles $\delta_j$ are unique by assuming that $\delta_1,\delta_2,\delta_3,\delta_4,\delta_5,\delta_6\in[0,\pi)$. If $\delta_3\in\{0,\frac\pi2\}$ then $\delta_1=0$, if $\delta_4\in\{0,\frac\pi2\}$ then $\delta_2=0$, if $\delta_4=0$ then $\delta_3=0$, and if $\delta_6\in\{0,\frac\pi2\}$ then $\delta_5=0$. The existence and uniqueness of such $\beta_j$ and $\delta_j$ are given by Lemmas 10 and 11 of~\cite{clement2022lov}.}
    \end{figure}

\begin{definition}[\lopp-calculus]
     Two \lopp-circuits $D$, $D'$ are equivalent according to the rules of the \lopp-calculus, denoted $\lopp\vdash D=D'$, if one can transform $D$ into $D'$ using the equations given in \pref{fig:nLOpp}. More precisely, $\lopp \vdash \cdot=\cdot$ is defined as the smallest congruence which satisfies the equations of \pref{fig:nLOpp} and the axioms of PROP.
\end{definition}

\begin{proposition}[Soundness of \lopp]\label{prop:soundLOpp}
    For any \lopp-circuits $D$ and $D'$, if $\lopp\vdash D=D'$ then $\interpsone{D}=\interpsone{D'}$.
\end{proposition}

\begin{proof}
    Since semantic equality is a congruence, it suffices to check that for every equation of \pref{fig:nLOpp}. The soundness of Equations~(\ref{nLOpp:swap}), (\ref{nLOpp:Rp2pi}) and (\ref{nLOpp:pM}) are direct consequences of \pref{def:semLOpp}. We can notice that $R_X(-2\theta)=\interpsone{\input{bs-xs.tikz}}$ and $e^{i\frac{\varphi}{2}}R_Z(\varphi)=\interpsone{\input{idtensps-xs.tikz}}$, where $R_X$ (resp.\ $R_Z$) is the rotation operator about the $\hat{x}$ axis (resp.\ $\hat{z}$ axis) of the Bloch sphere~\cite{nielsen2010quantum}. Any unitary of $U(2)$ can be decomposed into $e^{i\cdot}R_X(\cdot)R_Z(\cdot)R_X(\cdot)$ (resp.\ $e^{i\cdot}R_Z(\cdot)R_X(\cdot)R_Z(\cdot)$), giving the LHS (resp.\ RHS) angles of (\ref{nLOpp:E2}). More explicit formulae are given in \pref{subsec:E2angles}. By transforming the $iY$-axis into the $Y$-axis, the three rotations of the LHS (resp.\ RHS) of (\ref{nLOpp:E3}) can be seen as three real rotations along the $z-x-z$ real axes (resp. $x-z-x$). The angles are therefore given by the Euler angles~\cite{goldstein2002classical}. More explicit formulae are given in \pref{subsec:E3angles}.
\end{proof}
\begin{theorem}[Completeness of \lopp] \label{thm:compLOpp}
    For any two \lopp-circuits $D$ and $D'$, if $\interpsone{D}=\interpsone{D'}$ then $\lopp\vdash D=D'$.
    
\end{theorem}

\begin{proof}
    The equational theory of~\pref{fig:oLOpp} has been proven to be complete in~\cite{clement2022lov}. All equations of \pref{fig:oLOpp} are derived by those of \pref{fig:nLOpp} in \pref{app:proofcompLOpp}.
\end{proof}

\begin{theorem}[Minimality]
    The equational theory of~\pref{fig:nLOpp} is minimal for \lopp-circuits, i.e. none of its equations can be derived from the others.
\end{theorem}

\begin{proof}
    We define an alternative interpretation which satisfies all the equations aside from the one we prove to be necessary. In particular: \begin{itemize}
        \item (\ref{nLOpp:Rp2pi}) is the only rule on one wire that changes the sum of the phases.
        \item (\ref{nLOpp:pM}) is the only rule on one wire that can reduce the number of phases different from $2\pi$.
        \item (\ref{nLOpp:swap}) is the only rule changing the parity of the number of SWAPs. 
        \item (\ref{nLOpp:E2}) is the only rule changing the parity of (number of beam splitter + number of SWAPs).\end{itemize}
    For (\ref{nLOpp:E3}), the proof is given in~\ref{proof:minimality}.
\end{proof}

\subsection{Useful triangular forms}
\label{subsec:triangles}
In this subsection, we introduce three classes of \lopp-circuits, with the following inclusions: $\Tmnrec\subset\Tmn\subset\T$. Their shape and properties are illustrated and summarized in \pref{tab:triangles}. They are of particular interest as \T-circuits are the normal forms of the \lopp-calculus~\cite{clement2022lov}, \Tmn-circuits will be used in the normal forms of the \love-calculus (\pref{def:NF}), and their uniqueness will be proved thanks to use of \Tmnrec-circuits (\pref{sec:completeness}).

\begin{definition}[\T-circuits] \label{def:triangular}
    A \T-circuit is a \lopp-circuit with the following shape: 
    \[\!\!\!\!\tikzfigbox{0.6}{trianglepoint8-c}\]
    with $\varphi_{i,j} \in [0,2\pi),\theta_{i,j}\in [0,\frac{\pi}{2}]$ and the following conditions: $\theta_{i,j}=0\Rightarrow (\forall j'>j, \varphi_{i,j'}=\theta_{i,j'}=0)$ and $\theta_{i,j}=\frac{\pi}{2}\Rightarrow \varphi_{i,j}=0$. $\theta_{i,j}$ is on the $i^{th}$ right (resp.\ $j^{th}$ left) diagonal, and on the $(i+j-1)^{th}$ row of beam splitters.
\end{definition}

\begin{remark}[Coefficients of \interpsone{\T}]
    \label{rem:coefpath}
    The coefficient $t_{i,j}$ of $\interpsone{\triangleone}$ is determined by the sum of all the paths from the $j^{th}$ input wire to the $i^{th}$ output wire, where for each path we multiply by a $\cos$ (resp.\ $\sin$) term when the photon is reflected on (resp.\ transmitted through) a beam splitter, and by a phase when the path crosses a phase shifter. For instance, $t_{1,2}=\cos(\theta_{1,2})e^{i\varphi_{1,2}}i\sin(\theta_{1,1})e^{i\varphi_{1,1}}$. More generally, we have $t_{i,j}=e^{i\varphi_{i,j}}\cos(\theta_{i,j}) \times q_{i,j} + r_{i,j}$ where $q_{i,j},r_{i,j}$ are terms depending uniquely on the angles with lower indexes.
\end{remark}

\begin{proposition}[Universality and Uniqueness of $T$]\label{prop:uniqT}
    For any \lopp-circuit $D$, there exists a unique circuit $T$ in triangular form of \pref{def:triangular} such that $\interpsone{D}=\interpsone{T}$.
\end{proposition}

\begin{proof}
     The existence is a direct consequence of~\cite{reck1994experimental} or the Proposition~26 of~\cite{clement2022lov}. The uniqueness is a consequence of \pref{rem:coefpath} by sequentially showing the uniqueness of $(\varphi_{i,j},\theta_{i,j})$ in $t_{i,j}$, and by noticing that for any $z$ with $0<|z|\leq 1$, there are unique $\varphi,\theta\in [0,2\pi)\times[0,\frac{\pi}{2})$ such that $e^{i\varphi}\costheta=z$, with $\varphi,\theta=(0,0)$ for the special case of $z=0$. More details are provided in~\pref{app:uniqTproof}.
\end{proof}
\begin{remark}
    A generic \T-circuit $T:n\rightarrow n$ has $\frac{n(n-1)}{2}$ beam splitters and $\frac{n(n+1)}{2}$ phase shifters, having a total of $n^2$ different angles, the dimension of the unitary group $U(n)$.
\end{remark}

\begin{definition}[$\Tmn$-circuits]\label{def:subtriangular}
    A \lopp-circuit $\Ntriangle:n+\ntilde\rightarrow m+\mtilde$ is a $\Tmn$-circuit if:
    \begin{enumerate}
        \item\label{propTt:lopp} $\Ntriangle$ is a \T-circuit as defined in \pref{def:triangular},
        \item\label{propTt:absinput} there is no beam splitter or phase shifter fully and directly connected to the $\ntilde$ last input wires, ie.\ $\varphi_{i,j}=\theta_{i,j}=0$ if $row_{i,j}=i+j-1>n$ and there doesn't exist $(k,\ell)$ such that $k+\ell-1=row_{i,j}-1$, $k<i$ and $\theta_{k,\ell}\neq 0$,
        \item\label{propTt:absoutput} there is no beam splitter or phase shifter fully and directly connected to the $\mtilde$ last output wires, ie.\ $\varphi_{i,j}=\theta_{i,j}=0$ if $row_{i,j}=i+j-1>m$ and there doesn't exist $(k,\ell)$ such that $k+\ell-1=row_{i,j}-1$, $k\geq i$ and $\theta_{k,\ell}\neq 0$, and
        \item\label{propTt:noid} there exists one nonzero $\theta_{i,j}$ for each of the last $\max(\ntilde,\mtilde)$ rows. 
    \end{enumerate}
\end{definition}

The Property~\ref{propTt:noid} is an additional constraint that appears in the normal forms defined in~\pref{def:NF}. Property~\ref{propTt:absinput} and~\ref{propTt:absoutput} imply the only nonzero angles have indexes $(i\leq m,j\leq n)$, leading to the following proposition, direct consequence of~\pref{rem:coefpath} and the proof of~\pref{prop:uniqT}.

\begin{proposition}[Uniqueness of $\Tmn$-circuits on their $m\times n$ submatrix]\label{prop:uniqSubT}
    For any $\Tmn$-circuits $\triangleone,\triangletwo:n+\ntilde\rightarrow m+\mtilde$, if $\submatrix{\interpsone{\triangleone}}{m}{n}=\submatrix{\interpsone{\triangletwo}}{m}{n}$ then $\triangleone=\triangletwo$, where $\submatrix{M}{k}{\ell}$ is the $k\times\ell$ matrix composed of the first $k$ rows and $\ell$ columns of $M$.
\end{proposition}

\begin{definition}[$\Tmnrec$-circuits] \label{def:Trec}
    A $\Tmn$-circuit $\Ntriangle:n+\ntilde\rightarrow m+\mtilde$ is a \Tmnrec-circuit if $\mtilde=n$.
\end{definition}

\begin{remark}
    As $\mtilde=n$, \Tmnrec-circuits have exactly $\ntilde\times n$ beam splitters shaped like in \pref{tab:triangles}. Furthermore, their angles are necessarily nonzero, as one zero would imply the rest of the right-diagonal to be zero with~\pref{def:triangular}, contradicting the Property~\ref{propTt:noid}. That particular shape and those nonzero properties will be useful in the proofs of~\pref{sec:completeness}.
\end{remark}


\section{\lovetitle-calculus}\label{sec:love}

\subsection{Fock space}\label{subsec:vectorspace}

As described in \pref{subsec:singlephoton}, the state space of one photon with $m$ spatial modes is $\mathbb{C}^m$, as it can be on a superposition of all the different positions. Photons are particles that can bunch and share the same state, so each mode can be occupied by many photons. Furthermore, to observe quantum effects like interferences, we need the photons to be indistinguishable, meaning we can't know \textit{which photon is in which state}. 

For those two reasons, the usual way to represent a state with indistinguishable photons is by using the occupation number representation, where we indicate \textit{``how many photons are there in each state''}. We consider the basis states $\Phi_m\coloneqq\{ \ket{n_1,n_2,\dots,n_m}, (n_1,n_2,\dots,n_m) \in \N^m \}$ \cite{aaronson2011computational}, denoted as \textit{Fock states}. The state $\ket{n_1,n_2,\dots,n_m}$ represents a configuration where $n_i$ photons occupy the $i^{th}$ mode. The space of possible many-photon states over $m$ modes, the \textit{bosonic (symmetrical)} Fock space and denoted as $\Fbos{m}$, is defined as follows.

\begin{definition}[Fock space]\label{def:Fbos}
    We define the Hilbert space $\Fbos{m}$ as the Hilbert completion $\ell^2(\Phi_m)$ equipped with the bra-ket inner product $\langle \cdot | \cdot \rangle$, with the convention $\Fbos{0}=\C$.
\end{definition}

\begin{remark} $\Fbos{1}$ contains states that are an \textit{infinite} superposition of basis states, like the coherent states $\ket{\alpha}=e^{-\frac{|\alpha|^2}{2}}\sum_{k=0}^{\infty}\frac{\alpha^k}{k!}\ket{k}$. We can note that \Fbos{m} is isomorphic to $\ell^2(\N^m)$.
    
\end{remark}
To describe the space of the auxiliary sources, we consider a sub vector space of $\Fbos{m}$.

\begin{definition}[Subspace of the Fock space: $\Fv{m}$]\label{def:Fv}
    We define the pre-Hilbert space $\Fv{m}$ as the linear span of $\Phi_m$ equipped with the bra-ket inner product $\langle \cdot | \cdot \rangle$, with the convention $\Fv{0}=\C$.
\end{definition}

\begin{remark}
    \Fv{m} only contain states that are \textit{finite} linear combination of the Fock basis states. In particular, the coherent states are not included. We can note that \Fv{1} is isomorphic to $c_{00}$, i.e.\ the space of eventually zero sequences.
\end{remark}

\begin{definition}[$\Fvdual{\mtilde}$]
    We define the pre-Hilbert space $\Fvdual{\mtilde}$ as the linear span of $\{\bra{n_1,\dots,n_{\mtilde}},(n_1,\dots,n_{\mtilde})\in\N^{\mtilde}\}$.
\end{definition}

\subsection{Syntax and many-photon semantics}
\label{subsec:manyphotons}

\begin{definition}[\love-calculus]
    \love~is the PROP of \love-circuits generated by 
    \[\begin{array}{cccc}
        \tikzfigbox{1}{ps}: 1\rightarrow 1&
        \tikzfigbox{1}{bs}: 2\rightarrow 2&
        \input{sfboldntilde.tikz}: 0 \rightarrow \ntilde&
        \input{dgboldmtilde.tikz}: \mtilde \rightarrow 0
    \end{array}\] where $\varphi,\theta \in \mathbb{R}$, and  $\bm{f}$ (resp.\ $\bm{g}$) is a state in $\Fv{\ntilde}$ (resp.\ in $\Fvdual{\mtilde}$) with $\ntilde,\mtilde \in \N^+$. 
\end{definition}

\begin{remark}
    In string diagrams, a process $0\rightarrow \ntilde$ (resp.\ $\mtilde \rightarrow 0$) is called a state (resp.\ an effect). We will keep the source (resp.\ detector) terms to be consistent with their physical representation. A process $0\rightarrow 0$ is called a scalar.
\end{remark}

\begin{remark}
    The choice of those generators is discussed in \pref{app:notsobad}.
    
\end{remark}
\begin{definition}[Conventions for the notations]\label{def:conventions}
    Bold terms will always be vectors. In particular $\vect{f},\vect{f}_{k}$ (resp.\ $\vect{g},\vect{g}_{\ell}$) will always represent a ket (resp. a bra). Bold integers $\vect{k}$ (resp. \vect{\ell}) will represent $\ket{k}$ (resp.\ $\bra{\ell}$) in the sources (resp.\ detectors). The summation term $\sum$ will often be omitted, the index being implicitly the sum index. Note that for clarity, the summation term won't be omitted in~\pref{fig:rewritelofi}, and for conciseness, they will be omitted in~\pref{fig:axioms}. For instance $\vect{f}=\sum_{k\in\K}\ket{\vect{f}_k} \ket{k}$ will simply be noted as $\vect{f}_k \otimes \vect{k}$. \anyket~(resp.\ \anybra) represents an arbitrary ket (resp.\ bra) on one mode.~\anyketv~(resp.\ \anybrav) represents an arbitrary ket (resp.\ bra) for an arbitrary number of modes, representing an arbitrary scalar when the number of modes is zero. Those are used to omit terms when the specific value of those terms are not of interest, as in some equations of \pref{fig:axioms}. For the zero vector $\vect{f}=0$ (resp.\ $\vect{g}=0$), as there is no term in the sum, we chose to represent it with $\tikzfigbox{1.0}{s0}$ (resp.\ an empty detector $\tikzfigbox{1.0}{d0}$). Note it is different from $\tikzfigbox{1.0}{sket0}$ (resp.\ $\tikzfigbox{1.0}{dbra0})$ representing the nonzero vector $\ket{0}$ (resp.\ $\bra{0}$).
\end{definition}

\begin{definition} \label{def:semmany}
    Let $C$: $n \to m$ a \love-circuit, let $\interps{C}:\Fbos{n}\to\Fbos{m}$ be the linear map inductively defined as 
$\interps{C_2 \circ C_1}=\interps{C_2} \circ \interps{C_1}$, $\interps{C_1 \otimes C_2}=\interps{C_1}\otimes\interps{C_2}$
and:
\[
\begin{array}{cll}
\interps{\input{sfboldntilde.tikz}} \qquad& 0 \to \Fbos{\ntilde} \qquad\qquad& \vect{f}\in\Fv{\ntilde} \vspace{\pspazer}\vspace{\pspazer}\\ 
\interps{\input{dgboldmtilde.tikz}} & \Fbos{\mtilde} \to 0 & \vect{g}\in\Fvdual{\mtilde} \vspace{\pspazer}\\
   \interps{\tikzfigbox{\bx}{id}}  & \Fbos{1} \to \Fbos{1} & \ket{k} \mapsto \ket{k}\\\vspace{\pspazer}
     \interps{\begin{tikzpicture}[scale=0.85]
	\begin{pgfonlayer}{nodelayer}
		\node [style=none] (5) at (-1.1, 0) {};
		\node [style=none] (9) at (1.1, 0) {};
		\node [style=phase] (11) at (0, 0) {\small $\varphi$};
	\end{pgfonlayer}
	\begin{pgfonlayer}{edgelayer}
		\draw (5.center) to (11);
		\draw (11) to (9.center);
	\end{pgfonlayer}
\end{tikzpicture}
}  & \Fbos{1} \to \Fbos{1} & \ket{k} \mapsto \Pphi(\ket{k})\\\vspace{\pspazer}
     \interps{\tikzfigbox{\bx}{swap}}& \Fbos{2}\to \Fbos{2}& \ket{k_1,k_2} \mapsto \ket{k_2,k_1}
     
     \\ \vspace{\pspazer}
     \interps{\input{bs-xs.tikz}}&: \Fbos{2}\to \Fbos{2}& \ket{k_1,k_2} \mapsto \Btheta(\ket{k_1,k_2}) 
\end{array}\]
where $\Btheta(\ket{k_1,k_2})\coloneqq \sum\limits_{\ell_1+\ell_2=k_1+k_2}\sqrt{\frac{\ell_1!\ell_2!}{k_1!k_2!}} \sum\limits_{\substack{p+q=\ell_1\\ \delta=p-q}}{\binom{k_1}{p}}{\binom{k_2}{q}}\costheta^{k_2+\delta}(i\sintheta)^{k_1-\delta} \ket{\ell_1,\ell_2}$ and \\$\Pphi(\ket{k})\coloneqq e^{ik\varphi}\ket{k}$, with the convention ${\binom{k}{k'}}=0$ for $k<k'$.
\end{definition}

We can check that \Pphi~and \Btheta~are unitary operators~\cite{aaronson2011computational} and are therefore well-defined on the Hilbert space by continuity and linearity. One can also look at~\cite{kok2007review} or the Examples~\ref{ex:adagps} and~\ref{ex:adagbs} in~\pref{app:decompDelta} for a more physical interpretation of where the formulas come from.

\begin{remark}[Hermitian conjugate]
    We have $P_{\varphi}^\dagger=P_{-\varphi}$ and $B_{\theta}^\dagger=B_{-\theta}$, where $\dagger$ is the Hermitian conjugate. Therefore, $\bra{\ell}P_{\varphi}=(P_{-\varphi}\ket{\ell})^\dagger$ and $\bra{\ell_1,\ell_2}B_{\theta}=(B_{-\theta}\ket{\ell_1,\ell_2})^\dagger$.
\end{remark}

\subsection{Equational theory of~\lovetitle}
\label{subsec:lofeq}

\newcommand{\axiom}[5][1.0]{
    \begin{equation}  \tag{#3} \label{axiom:#3} #2~\begin{array}{rcl} \scalebox{#1}{\input{#4.tikz}}\!\!\!\!\!&=&\!\!\!\!\!\scalebox{#1}{\input{#5.tikz}}\end{array}\end{equation}}

    \newcommand{\axioms}[5][1.0]{\begin{equation}  \tag{#3} \label{axiom:#3} #2~\begin{array}{rcl} \scalebox{#1}{\input{#4.tikz}}\!\!\!\!\!\!&=&\!\!\!\!\!\!\!\!\!\!\scalebox{#1}{\input{#5.tikz}}\end{array}\end{equation}}
    \newcommand{\axiomss}[5][1.0]{\begin{equation}  \tag{#3} \label{axiom:#3} #2~\begin{array}{rcl} \scalebox{#1}{\input{#4.tikz}}\!\!\!\!\!\!\!\!\!\!&=&\!\!\!\!\!\!\!\!\!\!\scalebox{#1}{\input{#5.tikz}}\end{array}\end{equation}}
    \newcommand{\axiomsss}[5][1.0]{\begin{equation}  \tag{#3} \label{axiom:#3} #2~\begin{array}{rcl} \scalebox{#1}{\input{#4.tikz}}\!\!\!\!\!&=&\!\!\!\!\!\scalebox{#1}{\input{#5.tikz}}\end{array}\end{equation}}

\begin{figure}[htbp]
    \centering
    \renewcommand{\nspazer}{-0em}
    \begin{multicols}{2}
    \begin{minipage}{0.4\textwidth}
        \begin{equation}\tag{p\ensuremath{2\pi}}\label{axiom:p2pi} \hspace{-0.5cm}\vspace{0.15cm}~\begin{array}{rcl} \scalebox{0.8}{\input{LOpp-A-middle.tikz}}\!\!\!\!\!\!&=&\!\!\scalebox{0.8}{\input{LOpp-A-right.tikz}}\end{array}\end{equation}
    \end{minipage}\begin{minipage}{0.6\textwidth}
        \axiom[0.8]{}{swap}{swap}{swapright}
    \end{minipage}
\end{multicols}
\vspace{-0.5cm}
\axiom[0.9]{\quad\qquad}{E2}{E2left}{E2right} 
\axiom[0.8]{\qquad}{E3}{E3left}{E3right}
\vspace{0.1cm}

\begin{multicols}{2}   
    \begin{minipage}{0.4\textwidth}
        \axiom[0.8]{}{s0-0d}{Rsd0left}{sempty} 
    \end{minipage}\begin{minipage}{0.6\textwidth}
        \axiom[0.8]{\quad\qquad}{zero}{s0left}{s0right} 
   \end{minipage}
\end{multicols}

        \begin{multicols}{2}
            \begin{minipage}{0.4\textwidth}
                \axiom[0.8]{\hspace{-0.5cm}}{ss}{sMleft}{sMright2}
            \end{minipage}\begin{minipage}{0.6\textwidth}
                \axioms[0.9]{\hspace{-0cm}}{s-b}{sAbleft2}{sAbright2}\end{minipage}
    \end{multicols}

    \begin{multicols}{2}
                \begin{minipage}{0.42\textwidth}
                    \axiomsss{\hspace{-0.63cm}}{s-0d}{Roleft2}{Roright}
                \end{minipage}\begin{minipage}{0.58\textwidth}\axioms{}{s-p}{sApleft2}{sApright2}
            \end{minipage}\end{multicols}
                \begin{multicols}{2}
                \begin{minipage}{0.4\textwidth}
                    \axiom[0.8]{\hspace{-0.5cm}}{dd}{dMleft}{dMright}
                \end{minipage}\begin{minipage}{0.6\textwidth}
                \axiomsss[0.9]{\hspace{0cm}}{b-d}{dAbleft}{dAbright}
            \end{minipage}\end{multicols}
           \begin{multicols}{2}
            \begin{minipage}{0.42\textwidth}
            \axiomsss{\hspace{-0.63cm}}{s0-d}{s0-dleft}{s0-dright}
            \end{minipage}\begin{minipage}{0.58\textwidth}
            \axiomsss{}{p-d}{dApleft}{dApright}\end{minipage}
    \end{multicols}

        \axiom{\quad}{h2}{hleft2}{hright2}

    \caption{\label{fig:axioms}Axioms of the \love-calculus. The angles of (\ref{axiom:E2}) and (\ref{axiom:E3}) are the same as in the axioms of the \lopp-calculus (\pref{fig:nLOpp}). $h$ is any linear function $\Fv{2}\rightarrow\Fv{2}$. The conventions for $\left\{\emptyset,\anyket,\anyketv,\anybra,\anybrav\right\}$, and the omitted sums are detailed in~\pref{def:conventions}. The interpretations of the axioms are given in~\pref{prop:soundness}.}
\end{figure}

Similarly to \pref{subsec:eqlopp}, we introduce an equational theory for \love~in \pref{fig:axioms}.

\begin{definition}[\love-calculus]
    \label{def:lov-calc}
    Two \love-circuits $C, C'$ are
    equivalent according to the rules of the \love-calculus,
    denoted $\love\vdash C=C'$, if one can transform $C$ into $C'$ using the equations given in Figure~\ref{fig:axioms}.
  \end{definition}

\begin{remark}
    The~\pref{nLOpp:pM} is not present for it can be derived with the Equations~(\ref{axiom:p2pi}), (\ref{axiom:E2}) and~(\ref{axiom:s0-0d}), alongside with~\pref{nLOpp:b0}, that can be derived with the rules of the PROP, and the Equations~(\ref{axiom:swap}) and~(\ref{axiom:E2}). The derivation is detailed in~\pref{app:loficompletelopp}. Note that the Equaiton~(\ref{axiom:h2}) can be considered an equation of \textit{diagrams with holes}.
\end{remark}

\begin{proposition}[Soundness]\label{prop:soundness}
    For any two \love-circuits $C$ and $C'$, if $\love\vdash C = C'$ then $\interps{C}=\interps{C'}$. 
\end{proposition} 
\begin{proof}
    Since semantic equality is a congruence, it suffices to check the soundness for every equation of \pref{fig:axioms}, which follows from \pref{prop:soundLOpp} and \pref{def:semmany}. Informally, \pref{axiom:zero} means that if there is the scalar\footnote{\begin{tikzpicture}
	\begin{pgfonlayer}{nodelayer}
		\node [style=gene] (0) at (-1.25, 0) {$\vect{0}$};
		\node [style=none] (1) at (1.25, 0) {};
		\node [style=detector] (2) at (1.25, 0) {\vect{1}};
	\end{pgfonlayer}
	\begin{pgfonlayer}{edgelayer}
		\draw (0) to (1.center);
	\end{pgfonlayer}
\end{tikzpicture}
 is an impossible event and is one way to represent the scalar $0=\bra{1}\ket{0}=\interps{\begin{tikzpicture}
	\begin{pgfonlayer}{nodelayer}
		\node [style=gene] (0) at (-1.25, 0) {$\vect{0}$};
		\node [style=none] (1) at (1.25, 0) {};
		\node [style=detector] (2) at (1.25, 0) {\vect{1}};
	\end{pgfonlayer}
	\begin{pgfonlayer}{edgelayer}
		\draw (0) to (1.center);
	\end{pgfonlayer}
\end{tikzpicture}
}$.} $0$, then the semantics of all the circuit ($X\otimes 0=0$) is the null function. We can therefore nullify the other wires with the zeros \tikzfigbox{1.0}{s0} and \tikzfigbox{1.0}{d0}. This rule is specifically used for \pref{rem:NFzero}.~\pref{axiom:s-0d} means we can either (from LHS to RHS) project on $\ket{0}$ on the last mode or (from right to left) add any states $\vect{f}_k\otimes\ket{k}$ with $k\neq 0$ as they are trivially orthogonal and cancelling. \pref{axiom:h2} means we can \textit{shift} any function $h:\Fv{2}\rightarrow \Fv{2}$ from left to right where there are identity wires, direct consequence of the associativity: $\bra{\ell_1,\ell_2}(h\ket{k_1,k_2})=(\bra{\ell_1,\ell_2}h)\ket{k_1,k_2}$. The rules~(\ref{axiom:dd}), (\ref{axiom:b-d}), (\ref{axiom:p-d}) and~(\ref{axiom:s0-d}) are respectively the duals of~(\ref{axiom:ss}), (\ref{axiom:s-b}), (\ref{axiom:s-p}) and~(\ref{axiom:s-0d}).
\end{proof}

\begin{theorem}[Completeness]\label{thm:completeness}
    For any two \love-circuits $C$ and $C'$, if $\interps{C}=\interps{C'}$ then $\love\vdash C = C'$.
\end{theorem}

\begin{proof}
The proof is in~\pref{subsec:completeLOfi}, direct consequence of the uniqueness of the normal forms of~\pref{sec:completeness}.
\end{proof}

\section{Unique normal forms leading to the completeness of the \lovetitle-calculus}

\label{sec:completeness}

\label{subsec:sketch}

We introduce a set of oriented rewriting rules in \pref{subsec:rewriting}, that converge to a set of \love-circuits with specific shape and properties, defined in~\pref{subsec:convergetoNF}. The proof of their uniqueness is summarised in \pref{subsec:uniqueNF}. As a direct collorary of the uniqueness of the normal forms, we prove the completeness of the \love-calculus in \pref{subsec:completeLOfi}.

\subsection{Deterministic rewriting procedure}
\label{subsec:rewriting}

A strongly normalising rewriting system, i.e.\ terminating to normal forms, has been introduced in~\cite{clement2022lov} for \lopp-circuits. We mainly reuse all the rules, alongside additional rules to now take into account the sources and the detectors.

\begin{definition}[Rewriting system]
    Our rewriting system is defined on the PRO\footnote{This is similar to~\cite{clement2022lov}, to prevent any deformation of the form $\input{PROPswapswap.tikz}=\begin{tikzpicture}[scale=0.7]
	\begin{pgfonlayer}{nodelayer}
		\node [style=none] (0) at (-1, 0.5) {};
		\node [style=none] (1) at (1, 0.5) {};
		\node [style=none] (2) at (-1, -0.5) {};
		\node [style=none] (3) at (1, -0.5) {};
	\end{pgfonlayer}
	\begin{pgfonlayer}{edgelayer}
		\draw (0.center) to (1.center);
		\draw (2.center) to (3.center);
	\end{pgfonlayer}
\end{tikzpicture}
$.} of \love-circuits with the rules of~\pref{fig:rewritelofi}. 
\end{definition}

\newcommand{\sizeRW}{0.9}
\begin{figure}[htbp]
    \begin{multicols}{2}
    \begin{minipage}{0.45\textwidth}
    \begin{align}
      \label{rewritephasemod2pi}\input{phasemod2pi-left.tikz}&\to\input{phasemod2pi-right.tikz}\\[1.5em]
        \label{rewritebsmod2pi}\tikzfigbox{\sizeRW}{bsmod2pi-left}&\to\tikzfigbox{\sizeRW}{bsmod2pi-right}\\[1.5em]
          \label{rewritefusionphaseshifts}\input{LOpp-D-left.tikz}&\to\input{LOpp-D-right.tikz}\\[1.5em]
             \label{rewritezerophaseshifts}\input{LOpp-A-left.tikz}&\to\input{LOpp-A-right.tikz}\\[1.5em]
             \label{rewritezerobs}\input{LOpp-B-left.tikz}&\to \input{LOpp-B-right.tikz}
      \end{align}
    \end{minipage}\begin{minipage}{0.55\textwidth}

      \begin{align}
      \label{rewritetopphase}\tikzfigbox{\sizeRW}{rewritetopphase-left}&\to \tikzfigbox{\sizeRW}{rewritetopphase-right}\\[0.6em]
      \label{rewritepisur2}\tikzfigbox{\sizeRW}{rewritepisur2-left}\ &\to\ 
    \tikzfigbox{\sizeRW}{rewritepisur2-right}\\[0.6em]
      \label{rewritethetabs}\tikzfigbox{\sizeRW}{rewritethetabs-left}\ &\to\ \tikzfigbox{\sizeRW}{rewritethetabs-right}\\[0.6em]
      \label{rewriteminuspi}
    \tikzfigbox{\sizeRW}{rewriteminuspi-left}\ &\to\ \tikzfigbox{\sizeRW}{rewriteminuspi-right}
      \end{align}
    \end{minipage}
    \end{multicols}
     \vspace{-0.5cm}
       \begin{equation}\label{rewriteE3}\begin{array}{rcl}\tikzfigbox{0.8}{rewriteE3-left}&\to&\tikzfigbox{0.8}{LOPP-G-right}\end{array}\end{equation}
       
     \vspace{-0.4cm}
    
      \begin{equation}  \label{rewriteE2}\qquad\qquad\begin{array}{rcl}\input{rewriteE2-left.tikz}&\to&\input{E2right.tikz}\end{array}\end{equation}

      \begin{equation}\label{rewritezerof}\begin{array}{rcl}\input{rewritezerof-left.tikz}&\to&\input{rewritezerof-right.tikz}\end{array}\end{equation}
      \begin{equation}\label{rewritezerog}\begin{array}{rcl}\input{rewritezerog-left.tikz}&\to&\input{rewritezerog-right.tikz}\end{array}\end{equation}
      \begin{equation}\label{rewriteremoveg}\begin{array}{rcl}\input{rewriteremoveg-sum-left.tikz}&\!\!\!\!\to\!\!\!\!&\input{rewriteremoveg-sum-right.tikz}\end{array}\end{equation}\vspace{-0.2cm}
      \begin{equation}\label{rewritewire}\begin{array}{rcl}\input{rewritewire-left.tikz}&\to&\input{rewritewire-right.tikz}\end{array}\end{equation}
    
      \caption{Rewriting system of the \love-calculus, alongside with the oriented version, from the LHS to the RHS, of the axioms~(\ref{axiom:ss}), (\ref{axiom:s-b}), (\ref{axiom:s-p}), (\ref{axiom:dd}), (\ref{axiom:b-d}), and (\ref{axiom:p-d}). 
        $\psi \in \mathbb R\setminus [0,2\pi)$, $ \varphi, \varphi_1, \varphi_2\in (0,2\pi)$, $\theta_4\in (\frac{\pi}{2}, \pi)$, $\theta_5\in [\pi,2\pi)$, $\theta, \theta_1, \theta_2, \theta_3 \in (0,\pi)$. 
      $\tikzfigbox{0.8}{phasestar}$ denotes either $\tikzfigbox{0.8}{phase}$ or $\tikzfigbox{0.8}{id}$. The angles of the RHS of (\ref{rewriteE2}) and (\ref{rewriteE3}) are either given by \pref{subsec:Eangles} or by \cite{clement2022lov}. $\Niso_m:\N\rightarrow\N^m$ is a bijection arbitrary chosen to be $\Nisoinv_{m}\coloneqq\Nisoinv_{2}\circ\Nisoinv_{m-1}$ for $m>2$, where  $\Nisoinv_{2}(\ell,\ell')\coloneqq\frac{1}{2}(\ell+\ell)(\ell+\ell'+1)+\ell'$ is the Cantor pairing function and $\Niso_1$ is the identity. By convention, the summation index is $k\in\mathcal{K}$ for the sources and $\ell\in\mathcal{L}$ for detectors, aside from the rule~(\ref{rewriteremoveg}) where the sum is explicit for clarity.
      \label{fig:rewritelofi}}
    \end{figure}


We can check that all the rules are sound, and have the following meaning:
\begin{itemize}\label{item:meaning}
    \item The rules~\ref{rewritephasemod2pi}-\ref{rewriteE2} are either the same or just slightly different from the rules described in~\cite{clement2022lov}. With those rules, the $\lopp^{PRO}$-circuits will converge to the triangular \T-circuits defined in~\pref{subsec:triangles}.
    \item The rule~\ref{rewritezerof} removes any vector $\ket{f_{k'}}\sotimes\ket{k'}$ in the sources that is trivially cancelling with the detector on the connected last wire, meaning that $\bra{g_{k'}}=0$, i.e. that $k'\notin\mathcal{L}$.
    \item The rule~\ref{rewritezerog} removes any $\bra{g_{\ell'}}\sotimes\bra{\ell'}$ in the detectors that is trivially cancelling with the source on the connected last wire, meaning that $\ket{f_{\ell'}}=0$, i.e. that $\ell'\notin\mathcal{K}$.
    \item The rule~\ref{rewriteremoveg} allows, without changing the semantics, to transfer the generic coefficients from the detectors to the sources. Specifically, any term of the form $\sum_{\ell}\xi_{\ell}\bra{\mathcal{N}_{\mtilde}(L)}\bra{\ell}$ will be rewritten to $\bra{\mathcal{N}_{\mtilde}(L)}\bra{L}$. The coefficients $\xi$ will be in the source, as $\ket{f_L}\ket{L}$ will be rewriten to $\left(\sum_{i\in\mathcal{K}}\xi_i \ket{f_i}\right)\ket{L}$. At the end and by repeating this rule, there won't be any \textit{degree of freedom} in the detectors, and $\vect{g}=\sum_{\ell\in\mathcal{L}}\bra{\mathcal{N}_{\mtilde}(\ell)}\bra{\ell}$. The condition $(\xi_L\neq 1)\lor(\exists\ell\neq L, \xi_{\ell}\neq 0)$ is there to ensure that the rule is only used once for each $L$, and only when it's necessary.
    \item The rule~\ref{rewritewire} uses the bijection $\mathcal{N}_2: \mathbb{N}\to \mathbb{N}^2$ to remove one identity wire, by just relabelling the indexes in the sources and detectors. Note that one identity wire will always remain at the end.
    \item The oriented rule (from left to right) coming from the axioms~(\ref{axiom:ss}) and~(\ref{axiom:dd}) merge all the sources and detectors together.
    \item The oriented rule (from left to right) coming from the axioms~(\ref{axiom:s-b}), (\ref{axiom:s-p}), (\ref{axiom:b-d}) and (\ref{axiom:p-d}) reduce the number of phase shifters and beam splitters as much as possible, by making them be \textit{absorbed} into the sources and detectors.
\end{itemize}

\begin{definition}[Inputs of the rewriting system] For convenience, the inputs of the rewriting procedures are \love-circuits with at least one identity wire connecting sources and generators, and where all the sources (resp.\ detectors) are on the bottom left (resp.\ right).\end{definition}
\begin{remark} Note that choice is not restrictive, as the identity wire can always be added with~\pref{axiom:s0-0d}, and the sources and detectors can be placed at the desired position, without changing the semantics, with the rules of PROP and by adding SWAPs.\end{remark}

\begin{lemma}\label{lem:derivrewrite}
If $C_1$ rewrites to $C_2$ using the rules of~\pref{fig:rewritelofi}, then $\love\vdash C_1=C_2$. 
\end{lemma}

\begin{proof} \label{proof:univNF}
The rules are derived in \pref{app:derivrewrite}.
\end{proof}

\subsection{Normal forms of the \lovetitle-calculus}
\label{subsec:convergetoNF}

Formally with the rules of~\pref{fig:rewritelofi} and informally with their meaning described in~\pref{subsec:rewriting}, we can show that an irreducible form is a \love-circuit defined as follows:

\newcommand{\NFdef}{\ensuremath{\left(\tikzfigbox{\bxs}{id}^{\otimes m}\otimes\left(\input{dgboldonemtilde.tikz}\right)\right)\circ (\Ttilde\otimes \tikzfigbox{\bxs}{id}) \circ \left(\tikzfigbox{\bxs}{id}^{\otimes n}\otimes\left(\input{sfboldonentilde.tikz}\right)\right)}}

\begin{definition}[Normal form]\label{def:NF}
    The normal forms of any nonzero \love-circuits are denoted $N(T,\vect{f}): n\rightarrow m$ and are of the form:

    \[\input{NFexpl-T.tikz} \]
where: 

\begin{itemize}
    \item $\vect{f}$ is a nonzero generic state of $\Fv{\ntilde+1}$. 
    \vspace{0.2cm}
    \item $\vect{g}=\sum_{\ell \in \K} \bra{\Niso_m(\ell)}\otimes \bra{\ell}$, where $\Niso_m:\N\rightarrow\N^m$ is a bijection defined in~\pref{fig:rewritelofi} and $\K$ is the nonempty finite set $\left\{k\in\N\mid \vect{f}_k\neq 0\right\}$ of $\vect{f}=\sum_{k\in\K}\vect{f}_k\otimes\ket{k}$, with the convention $\K=\{0\}$ if $\ntilde=0$ or $\mtilde=0$.
    \vspace{0.2cm}
    \item $T: n+\ntilde \rightarrow m+\mtilde$ is a $\Tmn$-circuit as defined in \pref{def:subtriangular}.
    
\end{itemize}
   
\end{definition}
    
\begin{remark}
    If $\ntilde=\mtilde=0$, then the normal form is a normal form of \lopp~(can be $\begin{tikzpicture}
	\begin{pgfonlayer}{nodelayer}
		\node [style=sdiagrammevide] (2) at (0, 0) {};
	\end{pgfonlayer}
\end{tikzpicture}
$ for $n=m=0$) tensored with the scalar \tikzfigbox{\bx}{scalarNF} which has the semantics of a global scalar $\alpha\in\C$.
\end{remark}

\begin{remark}\label{rem:NFzero}
    We could also consider the particular case of $\vect{f}=0$, i.e.\ $\K=\emptyset$, where $\interps{N}:\Fbos{n}\rightarrow \Fbos{m}$ is the null function. We show in \pref{app:NFzeroproof} that $N:n\rightarrow m$ can be written to $(\tikzfigbox{0.8}{s0})^{\otimes m}\circ(\tikzfigbox{0.8}{d0})^{\otimes n}$, which is a more fitted form for representing the null function.
\end{remark}

\begin{lemma}[Strongly normalising]\label{lem:strongly}
    The rewriting system of~\pref{fig:rewritelofi} is strongly normalising.
\end{lemma}
\begin{proof}
We introduce a ranking function $(x_1,\dots,x_6)\in\mathbb{N}^6$, where each component of the tuple is determined by properties of the circuit, like the number of beam splitter with angles out of $[0,2\pi)$, the number of sources and detectors, or the number of identity wires connecting them. One nontrivial component is $x_6$, that we explicit here.

Let note the generic terms in the sources as $\vect{f}=\sum\alpha_{k_1,\dots,k_{\ntilde+1}}\ket{k_1,\dots,k_{\ntilde+1}}$ and in the detectors as $\vect{g}=\sum\beta_{\ell_1,\dots,\ell_{\mtilde+1}}\ket{\ell_1,\dots,\ell_{\mtilde+1}}$. We define:
 \[x_6\coloneq\sum\limits_{\vect{f}\in\text{sources}}C_1(\vect{f})+\sum\limits_{\vect{g}\in\text{detectors}}\left(2C_2(\vect{g})-C_3(\vect{g})\right)\]
 
 with $C_1(\vect{f})\coloneq\#\{\alpha_{k_1,\dots,k_{\ntilde+1}}\neq 0\}$, $C_2(\vect{g})\coloneq\#\{\beta_{\ell_1,\dots,\ell_{\mtilde+1}}\neq 0\}$, and \\ $C_3(\vect{g})\coloneq\#\{\beta_{\Niso_{\mtilde}(L),L}= 1, L\in \mathbb{N}\}$. The proof to show that the rule~(\ref{rewriteremoveg}) strictly decreases $x_6$ is the following. Let us consider the two cases: $(\xi_L\neq 1) \wedge (\forall \ell\neq L, \xi_{\ell}=0)$ and~$(\xi_L= 1) \wedge (\exists \ell\neq L, \xi_{\ell}=0)$. The first case doesn't change $C_1$ and $C_2$, but the term $-C_3$ strictly decreases by $1$. The second case doesn't change $C_3$, and the increase of $C_1$, i.e.\ the amount of new terms in $\vect{f}$, is bounded by $\#\{\xi_i\neq 0,i\neq L\}$, the number of terms removed in $\vect{g}$, which is the exact decrease of $C_2$. Therefore, $C_1+2C_2$ decreases by at least $\#\{\xi_i\neq 0,i\neq L\}>0$. We can conclude that the rule~(\ref{rewriteremoveg}) strictly decreases $x_6$.
The rest of the proof is given in \pref{app:strongly}.\end{proof}

Now that the normal forms are well-defined, it remains to prove their uniqueness, which is the purpose of the~\pref{subsec:uniqueNF}.

\begin{lemma}[Uniqueness of the Normal Forms]\label{lem:uniqN}
    If two \love-circuits $N$ and $N'$ in normal forms are such that $\interps{N}=\interps{N'}$, then $N=N'$.
    \end{lemma}


\subsection{The normal forms are unique: sketch of the proof}
\label{subsec:uniqueNF}

Let $N(T,\vect{f})$ be a normal form. In order to prove the uniqueness of $T$ and $\vect{f}$, we proceed with the following steps.
\begin{enumerate}
    \item We first show that $T$ is unique.
    \item We introduce a set of operators $\Omega^{\vect{i},\vect{j}}(T)$, such that $\interps{N}=\sum_{\vect{i},\vect{j}}\omega_{\vect{i},\vect{j}}\Omega^{\vect{i},\vect{j}}(T)$. We show the $\omega_{\vect{i},\vect{j}}$ to be canonically and uniquely associated with the coefficients of $\vect{f}$.
    \item We introduce a set of operators $\Delta^{\vect{u},\vect{v}}(T)$, that have very convenient properties and that we show to be linearly independent. 
    \item We give an isomorphism between the $\Omega$ and $\Delta$ operators, therefore proving the linear independence of the $\Omega^{\vect{i},\vect{j}}(D)$, and proving the uniqueness of the coefficients of $\vect{f}$.
\end{enumerate}

\begin{lemma}[Uniqueness of $T$]\label{lem:uniqT} For any two normal forms $N(T,\vect{f})$ and $N'(T',\vect{f'})$, if $\interps{N}=\interps{N'}$ then $T=T'$. 
\end{lemma}
\begin{proof}For any nonzero $W=\interps{\tikzfigbox{0.9}{basecaseleft}}$ and $W'=\interps{\tikzfigbox{0.9}{basecaseright}}$, we first show that: 
\[ (\theta,\varphi)\neq (\theta',\varphi')\Rightarrow \exists k\in\mathbb{N},\lim_{n\to\infty}\frac{\bra{n+k}W\ket{n}}{\bra{n+k}W'\ket{n}}\neq 1.\]
As $W=W'$, the limit of the ratio is necessarily equal to $1$; the parameters can't be different and are therefore equal. The proof relies purely on the semantics defined in \pref{def:semmany}. We then prove the uniqueness of $D$ by induction on the $\min{(\text{number of inputs},\text{number of outputs})}$, in~\pref{app:uniqT}. \end{proof}

\begin{definition}[$\Omega$ and $\Delta$ morphisms]\label{def:OmegaDelta}
    For any \lopp-circuit $D$, $(\vect{i},\vect{j},\vect{u},\vect{v})\in (\N^{\ntilde},\N^{\mtilde},\N^{\mtilde},\N^m)$ we define $\Omegav{i}{j}(D):n\rightarrow m$ and $\Deltav{u}{v}(D)$ as:
\[\begin{array}{lr} \Omegav{i}{j}(D)\coloneqq \interpspre{\input{Omegafigure.tikz}} & \Deltav{u}{v}(D)\coloneqq \interpspre{\input{Deltafigure.tikz}}\end{array}\]
where $\adag{}:\ket{k}\mapsto\sqrt{k+1}\ket{k+1}$ is the creation operator and $\interpspre{\cdot}\coloneqq\interps{\cdot}\vert_{pre}$ is the restriction of $\interps{\cdot}$ to $\Fv{}$.\end{definition} 

\begin{remark}
    The utility and the links between $\adag{}$ and $\interpspre{.}$ are described in~\pref{app:decompDelta}.
\end{remark}
\begin{remark}
    All the proofs regarding the $\Omega$ and $\Delta$ morphisms only consider the semantics on $\interpspre{.}$. That ensures the soundness of the proofs involving the unbounded operator $\adag{}$, as now all sums will be finite.
\end{remark}

We give here two propositions that are the core of the proofs.
\begin{proposition}[Unique $\Omega$-decomposition of the normal forms]\label{prop:uniqOmegaNF} For any \Tmn-circuit $T:n+\ntilde\rightarrow m+\mtilde$ and finite set $\left\{\omegav{i}{j}, (\vect{i},\vect{j})\in(\N^{\ntilde},\N^{\mtilde})\right\}$, there exists an unique normal form $N(T,\vect{f}):n\rightarrow m$, such that $\interpspre{N}=\sum_{\vect{i},\vect{j}\in(\setindex{I},\setindex{J})}\omegav{i}{j} \Omegav{i}{j}(T)$.
    \end{proposition}
    \begin{proof}
    It follows from the linearity of \interpspre{.} and that $\omega_{\vect{i},\vect{j}}=\input{omegascalarfromf.tikz}$.
    \end{proof}
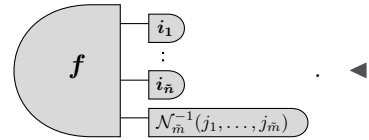
\begin{proposition}[Threshold properties of the $\Delta$-morphisms]\label{prop:deltathreshold}For any $\Tmnrec$-circuit $\Ndiamond:n+m\rightarrow n+m$ and $(\vect{u},\vect{v})\in(\N^n,\N^m)$, $\bra{\vect{y}}\Delta^{\vect{u},\vect{v}}(\Ndiamond)\ket{\vect{x}}$ is nonzero for $(\vect{x},\vect{y})=(\vect{v},\vect{u})$ and is zero if\\ $(\vect{x}\lrev \vect{v})\lor (\vect{y}\lrev\vect{u})$, where $\lrev$ is the reverse lexicographical order, i.e.\ $\vect{y}\lrev\vect{v}$ if there exists $k$ such that $y_n=v_n,\dots,y_{k+1}=v_{k+1}$ and $y_k<v_k$.\end{proposition}
\begin{proof}
    It is a consequence of the shape of the $\Tmnrec$-circuits (\pref{def:Trec}), and the properties of $\Delta^{\vect{u},\vect{v}}$. As there is no photon in the auxiliary sources, the input needs a certain number of photons for them to be detected in the auxiliary detectors. Similarly, as we create photons at the output with the creation operators $\adag{}$, the output needs a certain number of photons. More details are in \pref{app:uniqfullproof}.
\end{proof}

The linear independence of $\Delta$ will be a consequence of \pref{prop:deltathreshold} and a decomposition of the $\Omega$ with $\Delta$ morphisms will give the independence of the $\Omega$, thus the uniqueness of the $\omega_{\vect{i},\vect{j}}$, and therefore the uniquess of the normal forms with \pref{prop:uniqOmegaNF}. Everything is detailed in~\pref{app:uniqfullproof}.

\subsection{Completeness of the \lovetitle-calculus: Proof of Theorem~\ref{thm:completeness}}\label{subsec:completeLOfi}

Let $C,C'$ two \love-circuits such that $\interps{C}=\interps{C'}$. They can be rewritten to normal forms by \pref{lem:derivrewrite}: $\love\vdash C=N$ and $\love\vdash C'=N'$. By soundness of~\love, we have $\interps{N}=\interps{C}=\interps{C'}=\interps{N'}$ thus $\interps{N}=\interps{N'}$. By~\pref{lem:derivrewrite}, the normal forms are unique. Therefore, $N=N'$ and we have $\love\vdash~C=N=N'=C'$, thus $\love\vdash C=C'$, proving the completeness of the \love-calculus.\qed

\section{Outlook}

The formalism of the \love-calculus helped to find normal forms for linear optical circuits, and the new operators introduced in Section~\ref{sec:completeness} were particularly relevant for proving their uniqueness. It is an open problem to know if those normal forms and new operators can have further applications in simulation, compilation or the synthesis of linear optical circuits, or even broader reach as the \lopp-calculus had for quantum circuits~\cite{clement2023complete}. As those normal forms make only sense with finite states, it is also an open problem to determine whether normal forms exist in the infinite case, let alone their uniqueness.

\section*{Acknowledgements}
We would like to thank Marc de Visme and Vladimir Zamdzhiev for helpful discussions, Alexandre Clément for the insight into the angles of~(\ref{nLOpp:E3}) and the derivation of~(\ref{oLOpp:E3}), and particularly Shane Mansfield, Benoît Valiron and Renaud Vilmart for helpful discussions, support and reviews of some parts of the paper. We would also like to thank all the anonymous reviewers for their insightful comments and suggestions, which greatly helped to improve the quality of this manuscrip.

\clearpage
\bibliographystyle{plainurl} 
\bibliography{reference}

\appendix

\clearpage




\section{Notations}
\begin{table}[h]
    \centering
    \begin{tabular}{c|l}
        \textbf{Symbol} & \textbf{Meaning} \\
        \hline
        $C,C'$ & \love-circuits. \\
        $D,D',T,\Ntriangle,\Ndiamond$ & \lopp-circuits, cf \pref{tab:triangles} for the specific classes of circuits. \\ 
        $\varphi,\theta$ & Parameters (angles) of phase shifters and beam splitters\\
        $n,m,\ntilde,\mtilde$ & Integers used for the number of inputs ($n$ or $n+\ntilde$) and outputs ($m$ or $m+\mtilde$) \\
        $i,j,k,\ell,p,q$ & Integers used for indexing. \\
        $\vect{s},\vect{t},\vect{u},\vect{v},\vect{x},\vect{y}$ & Fock basis vectors.\\
        $\setindex{S,T,U,V}$ & Finite set of indexes associated with their lowercase vector. Often omitted in the sums. \\
        $\Fbos{m}$ & Hilbert space of the bosonic Fock space over $m$ modes, cf \pref{def:Fbos}. \\
        $\Fv{m} $ & Pre-Hilbert space of the bosonic Fock space over $m$ modes, cf \pref{def:Fv}. \\
        $\vect{f},\vect{f}'$ & Vectors of $\Fv{}$. \\
        $\vect{g},\vect{g}'$ & Vectors of $(\Fv{})^*$, the dual of \Fv{}. \\
        $\opc{a}{j}$& Creation operator over the mode $j$, introduced in \pref{def:OmegaDelta} \\
        $\Lambdav{\cdot}$& Operator defined in \pref{prop:commutelambda}\\
        $\Omegav{\cdot}{\cdot},\Deltav{\cdot}{\cdot}$ & Operators defined in \pref{def:OmegaDelta}. \\
        $\sum,\prod$ & Finite sums and products when the upper bound or the set of indexes is omitted. \\
        $\lessv,\lesseqv$ & Weak order on vectors, cf \pref{lem:Deltadec}.\\
        $\prec,\preceq$  & Lexicographic order on vectors, cf \pref{lem:Deltadec}.\\
        $\lrev,\leqrev$ & Reverse lexicographic order on vectors, cf \pref{prop:deltathreshold}.\\
    \end{tabular}
    \caption{\label{tab:notations}Notations used in the paper.}
\end{table}

\section{Section~\ref{sec:loqc}}\label{app:lopp}

\subsection{Euler angles of (\ref{nLOpp:E2}) and (\ref{nLOpp:E3})}\label{subsec:Eangles}

\newcommand{\genericUtwo}{\ensuremath{
    \begin{pmatrix}
        u_{1,1} & u_{1,2}\\
        u_{2,1} & u_{2,2}
        \end{pmatrix}}}

        \newcommand{\genericUtwoprime}{\ensuremath{
            \begin{pmatrix}
                u_{1,1}' & u_{1,2}'\\
                u_{2,1}' & u_{2,2}'
                \end{pmatrix}}}

\newcommand{\genericUthree}{\ensuremath{
    \begin{pmatrix}
        u_{1,1} & u_{1,2} & u_{1,3} \\
        u_{2,1} & u_{2,2} & u_{2,3} \\
        u_{3,1} & u_{3,2} & u_{3,3} \\
    \end{pmatrix}
}}
\newcommand{\absz}[1]{\ensuremath{\lvert #1 \rvert}}

\newcommand{\rhsmatrix}{\ensuremath{
    \begin{pmatrix}
        e^{i\beta_0}\cos({\beta_2}) &  ie^{i(\beta_0+\beta_1)}\sin({\beta_2})\\
        ie^{i\beta_3}\sin({\beta_2}) & e^{i(\beta_1+\beta_3)}\cos({\beta_2})
        \end{pmatrix}}}

Let us define the following matrices: 

\[R_X(\theta) \coloneqq \begin{pmatrix}
    \cos\left(\frac{\theta}{2}\right) & -i\sin\left(\frac{\theta}{2}\right) \\
    -i\sin\left(\frac{\theta}{2}\right) & \cos\left(\frac{\theta}{2}\right)
    \end{pmatrix} \quad\quad  H \coloneqq \frac{1}{\sqrt{2}}\begin{pmatrix}
            1 & 1 \\
            1 & -1
            \end{pmatrix}  \]
\[ R_z(\varphi) \coloneqq \begin{pmatrix}
    \cos(\varphi) & -\sin(\varphi) & 0 \\
    \sin(\varphi) & \cos(\varphi) & 0 \\
    0 & 0 & 1
    \end{pmatrix}\quad\quad  R_x(\theta) \coloneqq \begin{pmatrix}
        1 & 0 & 0 \\
        0 & \cos(\theta) & -\sin(\theta) \\
        0 & \sin(\theta) & \cos(\theta)
        \end{pmatrix} \]  

\[ P_{iy} \coloneqq \begin{pmatrix}
    1 & 0 & 0 \\
    0 & i & 0 \\
    0 & 0 & 1
    \end{pmatrix} \]  

    with $HH^{\dagger}=I_2=P_{iy}P_{iy}^{\dagger}$ where $I_2$ is the $2\times 2$ identity matrix and $\dagger$ is the conjugate transpose.

Let us also note:

\[ B(\theta) = \begin{pmatrix}
    \cos(\theta) & i\sin(\theta) \\
    i\sin(\theta) &  \cos(\theta )
    \end{pmatrix}\quad\quad P(\varphi) =\begin{pmatrix}
        1 & 0 \\
        0 &  e^{i\varphi} \\
        \end{pmatrix}  \] 

with 
\[  B_{1,2}(\theta) \coloneqq \begin{pmatrix}
    \cos(\theta) & i\sin(\theta) & 0\\
    i\sin(\theta) &  \cos(\theta ) & 0 \\
    0&0&1 \\
    \end{pmatrix}\quad\quad  B_{2,3}(\theta) \coloneqq \begin{pmatrix}
        1&0&0\\
        0&\cos(\theta) & i\sin(\theta) \\
        0& i\sin(\theta) &  \cos(\theta )  \\
        \end{pmatrix}\]

\subsubsection{Euler angles of (\ref{nLOpp:E2})}
\label{subsec:E2angles}
Let us note $U_{E_2}\coloneqq \genericUtwo $ with $\interpsone{LHS}= U_{E_2} =\interpsone{RHS}$. It is straightforward to compute the elements of $U_{E_2}$ with either the angles of the LHS or the RHS with the definitions (cf \pref{def:semLOpp}). Therefore, it is sufficient to give the angles from the elements of $U_{E_2}$. In the following paragraphs, we first give one solution for the angles, then precise all the angles that satisfy the equations.

\paragraph*{From $U_{E_2}$ to RHS}

We have:

\[ \interpsone{RHS} = \rhsmatrix = \genericUtwo \]

The following is a solution:
\begin{equation}\label{eq:E2onesolutionRHS}
    \left\{
    \begin{aligned}
    \beta_0 &= arg(u_{1,1}), \\
    \beta_1 &= arg(u_{2,2})-arg(u_{2,1})+\frac{\pi}{2},\\
    \beta_2 &= acos\left(\absz{u_{1,1}}\right), \\
    \beta_3 &= arg(u_{2,1}) -\frac{\pi}{2}, \\
    \end{aligned}
    \right.
\end{equation}
We can note that:
\begin{itemize}
    \item if $\beta_2+=\pi$ then the RHS is the same with $\beta_0+=\pi$ and $\beta_3+=\pi$,
    \item if $\beta_2\mapsto \pi - \beta_2$ then the RHS is the same with $\beta_0+=\pi$ and $\beta_1+=\pi$ ,
    \item all the angles can be taken modulo $2\pi$ without changing the RHS.
\end{itemize}

The complete set of solutions for $\beta_0,\beta_1,\beta_2$ and $\beta_3$ includes all possible configurations that can be derived from the values of (\ref{eq:E2onesolutionRHS}) and the three transformations outlined above.

Note there is a unique set of angles such that $\beta_2 \in [0,\frac{\pi}{2}]$, $\beta_0,\beta_1,\beta_3\in[0,2\pi)$ satisfying $\beta_1=0$ if $\beta_2\in\{0,\frac{\pi}{2}\}$. In particular, if $\absz{u_{1,1}}=1$, then $\beta_2=0$, $\beta_1=0$, and if $\absz{u_{1,1}}=0$, then $\beta_2=\frac{\pi}{2}$ and $\beta_1=0$.

\paragraph*{From $U_{E_2}$ to LHS}

Let us first notice that:

\[HP(\varphi) H = e^{i\frac{\varphi}{2}}R_X(\varphi)=e^{i\frac{\varphi}{2}}B({-\frac{\varphi}{2}})\]

We can therefore come back to a form similar to the RHS:

\[\begin{array}{rcl}
    H\interpsone{LHS}H&=& e^{i\alpha_0}H\left(B(\alpha_1)P(\alpha_2-\alpha_0)B(\alpha_3)\right)H\\
    &=& e^{i\alpha_0}\left(HB(\alpha_1)H\right)\left(HP(\alpha_2-\alpha_0)H\right)\left(HB(\alpha_3)H\right)\\
    &=& e^{i\alpha_0}\left(e^{i\alpha_1}P(-2\alpha_1)\right)\left(e^{i\frac{\alpha_2-\alpha_0}{2}}B\left(-\frac{\alpha_2-\alpha_0}{2}\right)\right)\left(e^{i\alpha_3}P(-2\alpha_3)\right)\\ 
    &=& e^{i\left(\frac{\alpha_0+2\alpha_1+\alpha_2+2\alpha_3}{2}\right)}P(-2\alpha_1)B\left(\frac{\alpha_0-\alpha_2}{2}\right)P(-2\alpha_3) \\
\end{array}\]

By noting $\genericUtwoprime=HUH$, and from the previous paragraph we can infer the following solution:

\begin{equation*}
    \left\{
    \begin{aligned}
    \frac{\alpha_0+2\alpha_1+\alpha_2+2\alpha_3}{2}&= arg(u_{1,1}'), \\
    \frac{\alpha_0-\alpha_2}{2} &= acos\left(\absz{u_{1,1}'}\right), \\
    -2\alpha_3+\frac{\alpha_0+2\alpha_1+\alpha_2+2\alpha_3}{2} &= arg(u_{2,1}') -\frac{\pi}{2}, \\
    -2\alpha_1 &= arg(u_{2,2}')-arg(u_{2,1}')+\frac{\pi}{2}
    \end{aligned}
    \right.
\end{equation*}

Thus we have the following solution:
\begin{equation}\label{eq:E2onesolutionLHS}
    \left\{
    \begin{aligned}
    \alpha_0&= arg(u_{1,1}')+arg(u_{2,2}'), \\
    \alpha_1 &= \frac{arg(u_{2,1}')-arg(u_{2,2}')}{2}-\frac{\pi}{4}\\
    \alpha_2&= arg(u_{1,1}')+arg(u_{2,2}') - 2acos\left(\absz{u_{1,1}'}\right), \\
    \alpha_3&= acos\left(\absz{u_{1,1}'}\right)- \frac{arg(u_{2,1}')+arg(u_{2,2}')}{2} +\frac{\pi}{4}, \\
    \end{aligned}
    \right.
\end{equation}

We can note that:
\begin{itemize}
    \item if $\alpha_1+=\pi$ then the LHS is the same with $\alpha_0+=\pi$ and $\alpha_2+=\pi$,
    \item if $\alpha_3+=\pi$ then the LHS is the same with $\alpha_0+=\pi$ and $\alpha_2+=\pi$,
    \item all the angles can be taken modulo $2\pi$ without changing the LHS.
\end{itemize}

The complete set of solutions for $\alpha_0,\alpha_1,\alpha_2$ and $\alpha_3$ includes all possible configurations that can be derived from the values of (\ref{eq:E2onesolutionLHS}) and the three transformations outlined above.

Note there is a unique set of angles such that $\alpha_1\in[0,\frac{\pi}{2}),\alpha_3 \in [0,\pi)$, $\alpha_0,\alpha_2\in[0,2\pi)$ satisfying $\alpha_1=0$ if $\alpha_0-\alpha_2=0~\text{mod}~\pi$. In particular, if $\absz{u_{1,2}}=0$, then $\alpha_0-\alpha_2=0~\text{mod}~\pi$ and $\alpha_1=\alpha_3=0$, and if $\absz{u_{1,1}}=0$, then  $\alpha_0-\alpha_2=0~\text{mod}~\pi$, $\alpha_1=0$ and $\alpha_3=\frac{\pi}{2}$.

\subsubsection{Euler angles of (\ref{nLOpp:E3})}

\label{subsec:E3angles}

Note that $P_{iy}^{\dagger}B_{1,2}(\theta)P_{iy}=R_z(\theta)$ and $P_{iy}^{\dagger}B_{2,3}(\theta)P_{iy}=R_x(\theta)$ with $P_{iy}^{\dagger}P_{iy}=I_3$, so we have:
\[ \begin{array}{rcl}
    P_{iy}^{\dagger}\interpsone{LHS}P_{iy} & =&  P_{iy}^{\dagger}B_{1,2}(\gamma_1)B_{2,3}(\gamma_2)B_{1,2}(\gamma_3)P_{iy} \\ 
    &=&   \left(P_{iy}^{\dagger}B_{1,2}(\gamma_1)P_{iy}\right)\left(P_{iy}^{\dagger}B_{2,3}(\gamma_2)P_{iy}\right)\left(P_{iy}^{\dagger}B_{1,2}(\gamma_3)P_{iy}\right)  \\
    &=&R_z(\gamma_1)R_x(\gamma_2)R_z(\gamma_3) \\
\end{array}\]

and
\[ \begin{array}{rcl}
    P_{iy}^{\dagger}\interpsone{RHS}P_{iy} & = & P_{iy}^{\dagger}B_{2,3}(\delta_1)B_{1,2}(\delta_2)B_{2,3}(\delta_3)P_{iy} \\ 
    &=& \left(P_{iy}^{\dagger}B_{2,3}(\delta_1)P_{iy}\right)\left(P_{iy}^{\dagger}B_{1,2}(\delta_2)P_{iy}\right)\left(P_{iy}^{\dagger}B_{2,3}(\delta_3)P_{iy}\right)\\ 
    &=& R_x(\delta_1) R_z(\delta_2)R_x(\delta_3)\\
\end{array}
\]

Thus the angles of (\ref{nLOpp:E3}) are exactly the Euler angles for rotations in $\mathbb{R}^3$:

\[ \begin{array}{cccc}
    &\interpsone{LHS}&=&\interps{RHS}  \\
\Leftrightarrow &P_{iy}^{\dagger}\interpsone{LHS}P_{iy}&=&P_{iy}^{\dagger}\interps{RHS}P_{iy} \\
\Leftrightarrow & R_z(\gamma_3)R_x(\gamma_2)R_z(\gamma_1) &=&R_x(\delta_3) R_z(\delta_2)R_x(\delta_1) \\
\end{array} \]

One can look at~\cite{goldstein2002classical} for the theory. For the explicit formula, let us define the $3\times 3$ real matrix $R_{E_3}\coloneq \begin{pmatrix}
    r_{1,1} & r_{1,2} & r_{1,3} \\
    r_{2,1} & r_{2,2} & r_{2,3} \\
    r_{3,1} & r_{3,2} & r_{3,3} \\
\end{pmatrix}=P_{iy}^{\dagger}\interpsone{LHS}P_{iy}=P_{iy}^{\dagger}\interpsone{RHS}P_{iy}$. It is straightforward to compute the elements of $R_{E_3}$ with either the angles of the LHS or the RHS with the definitions (cf \pref{def:semLOpp}). Therefore, it is sufficient to give the angles from the elements of $R_{E_3}$.

 Here, we give the formula directly taken from~\cite{eberly2020euler}.

\paragraph*{From $R_{E_3}$ to RHS}

\begin{itemize}
    \item If $-1<r_{1,1}<1$ then $\delta_1= atan2(r_{1,3},-r_{1,2})$, $\delta_2=acos(r_{1,1})$ and $\delta_3= atan2(r_{3,1},r_{2,1})$.
    \item If $r_{1,1}=1$, then necessarily $\delta_2=0$, and the solutions are all the angles such that $\delta_1+\delta_3=atan2(r_{3,2},r_{3,3})$.
    \item If $r_{1,1}=-1$, then necessarily $\delta_2=\pi$, and the solutions are all the angles such that $\delta_1-\delta_3=atan2(r_{3,2},r_{3,3})$.
\end{itemize}
All those angles can be taken modulo $2\pi$.
\paragraph*{From $R_{E_3}$ to LHS}

\begin{itemize}
    \item If $-1<r_{3,3}<1$ then $\gamma_2=acos(r_{3,3}),~\gamma_3=atan2(r_{1,3},-r_{2,3})$ and $\gamma_1=atan2(r_{3,0},r_{3,2})$.
    \item If $r_{3,3}=1$, then necessarily $\gamma_2=0$, and the solutions are all the angles such that $\gamma_1+\gamma_3=atan2(-r_{1,2},r_{1,1})$.
    \item If $r_{3,3}=-1$, then necessarily $\gamma_2=\pi$, and the solutions are all the angles such that $\gamma_1-\gamma_3=atan2(-r_{1,2},r_{1,1})$.
\end{itemize}

All those angles can be taken modulo $2\pi$.

\subsection{Completeness of \lopp: Theorem~\ref{thm:compLOpp}}
\label{app:proofcompLOpp}

To prove \pref{nLOpp:p0}, we have:
\begin{longtable}{rcl}
\begin{tikzpicture}
	\begin{pgfonlayer}{nodelayer}
		\node [style=none] (5) at (-1.1, 0) {};
		\node [style=none] (9) at (1.1, 0) {};
		\node [style=phase] (11) at (0, 0) {$0$};
	\end{pgfonlayer}
	\begin{pgfonlayer}{edgelayer}
		\draw (5.center) to (11);
		\draw (11) to (9.center);
	\end{pgfonlayer}
\end{tikzpicture}
 & $\overset{(p2\pi)}{=}$ & \begin{tikzpicture}
	\begin{pgfonlayer}{nodelayer}
		\node [style=none] (0) at (-1.5, 0) {};
		\node [style=none] (1) at (1.7, 0) {};
		\node [style=phase] (2) at (-0.65, 0) {$0$};
		\node [style=phase] (3) at (0.7, 0) {$2\pi$};
	\end{pgfonlayer}
	\begin{pgfonlayer}{edgelayer}
		\draw (0.center) to (2);
		\draw (2) to (1.center);
	\end{pgfonlayer}
\end{tikzpicture}
 \\\\
 & \eqeqref{nLOpp:pM} & \begin{tikzpicture}
	\begin{pgfonlayer}{nodelayer}
		\node [style=none] (5) at (-1.1, 0) {};
		\node [style=none] (9) at (1.1, 0) {};
		\node [style=phase] (11) at (0, 0) {$2\pi$};
	\end{pgfonlayer}
	\begin{pgfonlayer}{edgelayer}
		\draw (5.center) to (11);
		\draw (11) to (9.center);
	\end{pgfonlayer}
\end{tikzpicture}
 \\\\
 & $\overset{(p2\pi)}{=}$ & \begin{tikzpicture}
	\begin{pgfonlayer}{nodelayer}
		\node [style=none] (0) at (-1.5, 0) {};
		\node [style=none] (1) at (1.5, 0) {};
	\end{pgfonlayer}
	\begin{pgfonlayer}{edgelayer}
		\draw [style=noire] (0.center) to (1.center);
	\end{pgfonlayer}
\end{tikzpicture}
 \\\\
\end{longtable}

To prove \pref{nLOpp:b0}, we have:
\[\begin{array}{rcl}
\input{LOpp-B-left.tikz} & \overset{PROP}{=} \input{derivLopp-b0-1.tikz} \\\\ & \eqeqref{nLOpp:swap}&\input{derivLopp-b0-2.tikz} \\\\
& \eqeqref{nLOpp:E2} & \input{derivLopp-b0-3.tikz} \\\\
& \overset{(\ref{nLOpp:swap}),PROP}{=} \input{LOpp-B-right.tikz} \\\\
\end{array}\]

To prove \pref{oLOpp:pp}, we have:
\begin{longtable}{rcl}
\input{LOpp-E-left.tikz} & \eqeqref{nLOpp:b0} & \input{derivLOpp-ppb1.tikz} \\\\
& \eqeqref{nLOpp:E2} & \input{derivLOpp-ppb2.tikz} \\\\
& \eqeqref{nLOpp:p0} & \input{LOpp-E-right.tikz} \\\\
\end{longtable}

\begin{lemma}[Useful \lopp~equations]\label{lem:usefulLOppeq}
   We can derive the three following equations in \lopp:
    \begin{equation}\label{eq:movepbl}
        \begin{array}{rcl}\qquad\qquad\qquad \input{movepbl-left.tikz}&= &\input{movepbl-right.tikz}\end{array}
    \end{equation}
    \begin{equation}\label{eq:moveptr}
        \begin{array}{rcl} \qquad\qquad\qquad\input{moveptr-left.tikz}&= &\input{moveptr-right.tikz}\end{array}
    \end{equation}
    \begin{equation}\label{eq:pswap}
        \begin{array}{rcl}\qquad\qquad\qquad\input{pswap-left.tikz} = \input{pswap-right.tikz}\end{array}
    \end{equation}
    If $\sin(\delta_1)\neq 0$, we can also derive:
    \begin{equation}\label{eq:E3swap}
        \begin{array}{rcl}\input{E3swap-left.tikz} = \input{E3swap-right.tikz}\end{array}
    \end{equation}
    with $\alpha=\atan(\frac{c_{\delta_1}}{s_{\delta_1}c_{\delta_2}})$.
\end{lemma}

\begin{proof}

    To derive \pref{eq:movepbl}, we have:
    \begin{longtable}{rcl}
        \input{movepbl-left.tikz} & \eqdeuxeqref{nLOpp:p0}{nLOpp:pM} & \input{movepbl1.tikz} \\\\
        & \eqeqref{oLOpp:pp} & \input{movepbl-right.tikz} \\\\
        \end{longtable}

        To derive \pref{eq:moveptr}, we have:
        \begin{longtable}{rcl}
            \input{moveptr-left.tikz} & \eqdeuxeqref{nLOpp:p0}{nLOpp:pM} & \input{moveptr1.tikz} \\\\
            & \eqeqref{oLOpp:pp} & \input{moveptr-right.tikz} \\\\
            \end{longtable}

        To derive \pref{eq:pswap}, we have:
    \begin{longtable}{rcl}
        \input{pswap-left.tikz} & \eqeqref{nLOpp:swap} & \input{pswap1.tikz} \\\\
        & $\overset{PROP}{=}$ & \input{pswap2.tikz} \\\\
        & \eqdeuxeqref{nLOpp:pM}{nLOpp:pM} & \input{pswap3.tikz} \\\\
        & \eqeqref{nLOpp:swap} & \input{pswap-right.tikz} \\\\
        \end{longtable}

For \pref{eq:E3swap}, by looking at the matrices of both sides of the equation we have:

\[  \interpsone{LHS}=\begin{pmatrix}
    * & * & * \\
    * & * & * \\
    * & * & -s_{\alpha}s_{\delta_1}c_{\delta_2}+c_{\alpha}c_{\delta_1} \\
    \end{pmatrix}   = \begin{pmatrix}
        * & * & * \\
        * & * & * \\
        * & * & 0 \\ 
        \end{pmatrix} = \interpsone{RHS} \]

Therefore, for $\delta_1\neq 0$, we have: \[-s_{\alpha}s_{\delta_1}c_{\delta_2}+c_{\alpha}c_{\delta_1}=0 \Leftrightarrow \tan(\alpha)=\frac{c_{\delta_1}}{s_{\delta_1}c_{\delta_2}} \Leftrightarrow~\alpha=\atan(\frac{c_{\delta_1}}{s_{\delta_1}c_{\delta_2}}).\]

\end{proof}

To derive \pref{oLOpp:E2}, we have:

\begin{longtable}{rcl}
    \input{LOpp-F-left.tikz} & \eqeqref{nLOpp:b0} & \input{derivLOpp-E21.tikz} \\\\
    & \eqeqref{nLOpp:E2} & \input{derivLOpp-E22.tikz} \\\\
    & \eqdeuxeqref{eq:movepbl}{nLOpp:pM} & \input{LOpp-F-right.tikz} \\\\
    \end{longtable}

To derive \pref{oLOpp:E3} if $\gamma_3 =0$, we have:
\begin{longtable}{rcl}
    \input{derivLOpp-E3zero0.tikz} & \eqeqref{nLOpp:b0} & \input{derivLOpp-E3zero1.tikz} \\\\
    & \eqeqref{nLOpp:E2} & \input{derivLOpp-E3zero2.tikz} \\\\
    & \eqdeuxeqref{nLOpp:p0}{nLOpp:b0} & \input{derivLOpp-E3zero3.tikz} \\\\
    \end{longtable}

To derive \pref{oLOpp:E3} if $\gamma_3 \neq 0$, we have:
\begin{longtable}{rcl}
    &&\input{derivLOpp-E30.tikz} \\\\
    &\eqtroiseqref{nLOpp:E2}{nLOpp:p0}{nLOpp:b0} & \input{derivLOpp-E31.tikz} \\\\
    & \eqeqref{eq:E3swap} & \input{derivLOpp-E32.tikz} \\\\
    & \eqeqref{nLOpp:E2} & \input{derivLOpp-E33.tikz} \\\\
    & \eqdeuxeqref{eq:moveptr}{nLOpp:pM} & \input{derivLOpp-E34.tikz} \\\\
    & \eqeqref{eq:pswap} & \input{derivLOpp-E35.tikz} \\\\
    & \eqeqref{nLOpp:E3} & \input{derivLOpp-E36.tikz} \\\\
    & \eqeqref{nLOpp:E2} & \input{derivLOpp-E37.tikz} \\\\
    & \eqtroiseqref{eq:movepbl}{nLOpp:pM}{nLOpp:p0} & \input{derivLOpp-E38.tikz} \\\\
    \end{longtable}

\subsection{Minimality of \lopp: Equation~(\ref{nLOpp:E3}) is necessary} 
\label{proof:minimality}

Here is the sketch of the proof:
\newcommand{\equiph}{\ensuremath{\sim_{\varphi}}}
\begin{itemize}
    \item We define an equivalence relation $\equiph$ on three-wire \lopp-circuits.
    \item We introduce a confluent rewriting procedure that is conserving the relation $\equiph$, and that is converging to normal forms.
    \item All the rules of the PROP, (\ref{nLOpp:p0}), (\ref{nLOpp:swap}), (\ref{nLOpp:pM}) and (\ref{nLOpp:E2}) also conserve the relation $\equiph$.
    \item We conclude that (\ref{nLOpp:E3}) is necessary, because the LHS and RHS are different normal forms, and therefore can't be transformed from one to the other without (\ref{nLOpp:E3}).
\end{itemize}

\begin{definition}[Minimal block decomposition]
    Let $D$ be a \lopp-circuit on three wires. We define the minimal block decomposition as the minimal partition of $D$ into sub-circuits $B_1,B_2,\dots,B_n$ such that $D$, up to deformation, is of the form:

    \[\input{Bblocks.tikz}\]
    
    where for each $i$, $B_i$ contains as many as phase shifters as possible that are not part of the $B_1,\dots,B_{i-1}$, and $n$ is minimal.
\end{definition}

Note that $B_1$ and $B_n$ can be identities or swap, and that $B_{n-1}$ can be a SWAP (but not the identity).

\begin{definition}[Equivalence relation]
    Let $D$,$D'$ two three-wire \lopp-circuits. We have $D\equiph D'$ if and only if their minimal block decomposition $B_1,\dots,B_n$ and $B_1',\dots,B_n'$ have the same number of blocks, and are semantically equal up to phases that can be exchanged between the blocks. 
\end{definition}
It means that $D\equiph D'$ if there exists $\varphi_1,\dots,\varphi_n,\psi_1,\dots,\psi_n$ such that the semantics of the $i^{th}$ dash box of the following figure: 

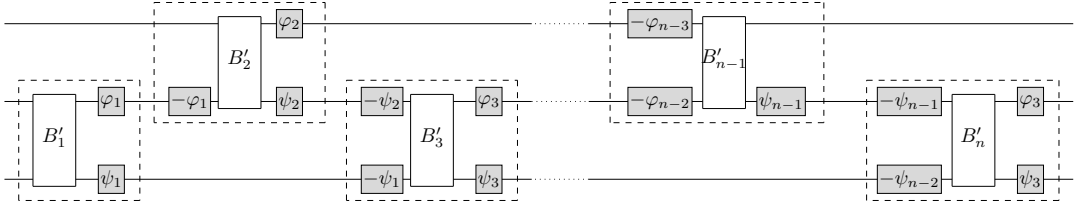
\[\tikzfigbox{0.75}{Bequivblocks}\]

is equal to $\interpsone{\input{Biblock.tikz}}$.

\begin{definition}[Rewriting system]
    We consider the rewriting system that is rewriting each block $B_i$ containing at least one beam splitter by \tikzfigbox{0.7}{phipsi} if $\interpsone{B_i}=\begin{pmatrix}
    e^{i\varphi}&0\\0&e^{i\psi}
    \end{pmatrix}$ and by \tikzfigbox{0.7}{phipsiswap} if $\interpsone{B_i}=\begin{pmatrix}
        0&e^{i\psi}\\e^{i\varphi}&0
        \end{pmatrix}$.
\end{definition}

We can notice that:
\begin{itemize}
    \item The system strictly decreases the number of beam splitters, hence is strongly normalising.
    \item It follows from the definition of \equiph that the systems preserves \equiph.
    \item The system is locally confluent. Starting from $D$, if we rewrite $B_i$ giving $D_1$ or $B_j$ giving $D_2$, we can independently rewrite $B_i$ in $D_2$ and $B_j$ in $D_1$, giving the same circuit $D'$.
\end{itemize}

Finally, we can notice that the rules (\ref{nLOpp:p0}), (\ref{nLOpp:swap}), (\ref{nLOpp:pM}), (\ref{nLOpp:E2}) and all the rules of the PROP also converse the relation $\approx_1$.

As the LHS and the RHS are distinct normal forms of our system, we know they can't be derived from the rules conserving~\equiph. Therefore, (\ref{nLOpp:E3}) is necessary.

\subsection{Properties of the triangular forms of Section~\ref{subsec:triangles}}

\renewcommand{\nspazer}{1em}
\begin{table}[h]
    \centering
    \begin{tabular}{|m{8cm}|m{5cm}|}
     \hline Shape & Properties \\\hline 
     & \T-circuits (\pref{def:triangular})  \\
     \tikzfigbox{0.5}{LOpp-T-s-new} \vspace{\nspazer} & 
     Uniquely determined by $\interpsone{\cdot}$ (\pref{prop:uniqT}).  \\
     
     \hline
     & \Tmn-circuits (\pref{def:subtriangular}) \\
     \tikzfigbox{0.5}{LOpp-T-34-new} \vspace{\nspazer} &  
     Uniquely determined by the submatrix $\submatrix{\interpsone{\cdot}}{m}{n}$ (\pref{prop:uniqSubT}). Used for the normal forms of \love. %
     \\

       \hline
      & \Tmnrec-circuits (\pref{def:Trec}) \\
      \vspace{\nspazer}\quad\quad\quad\quad\tikzfigbox{0.7}{LOpp-Trec-34-new} \vspace{\nspazer}& They have exactly $\ntilde \times \mtilde$ nonzero beam splitters, with no identity wire. They are used in the proofs of \pref{sec:completeness}.
     \\\hline
    \end{tabular}
    \caption{\label{tab:triangles}Shapes and properties of classes of triangle \lopp-circuits: $n+\ntilde\rightarrow m+\mtilde$. $(\anyphi,\anytheta)$ are angles in $[0,2\pi)\times [0,\frac{\pi}{2}]$ that satisfy the properties of Defitions~\ref{def:triangular} and~\ref{def:subtriangular}. We emphasis the nonzero angles of \Ndiamond~ by noting $\anynztheta$ an arbitrary angle in $(0,\frac{\pi}{2}]$. The angles which are necessarily zero for the property 3 and 4 of \pref{def:subtriangular} are in red. We have $n=\ntilde=3,m=4$ and $\mtilde=2$ for the first two figures, and $\mtilde=n=2$ and $\ntilde=m=3$ for the third.}
\end{table}

\subsubsection*{Proof of Proposition~\ref{prop:uniqT}}\label{app:uniqTproof}
The coefficient $t_{i,j}$ of $\interpsone{\triangleone}$ is determined by the sum of all the paths from the $j^{th}$ input wire to the $i^{th}$ output wire, where for each path, we multiply by a $\cos$ (resp.\ $\sin$) term when the photon is reflected on (resp.\ transmitted through) a beam splitter, and by a phase when the path crosses a phase shifter. For instance: \[\begin{array}{rcl} t_{1,2}&=&\cos(\theta_{1,2})e^{i\varphi_{1,2}}i\sin(\theta_{1,1})e^{i\varphi_{1,1}} \text{ and} \\ t_{2,2}&=&i\sin(\theta_{1,2})\cos(\theta_{2,2})e^{\varphi_{2,2}}i\sin(\theta_{2,1})e^{i\varphi_{2,1}}+\cos(\theta_{1,2})e^{\varphi_{1,2}}i\sin(\theta_{1,1})\cos(\theta_{2,1})e^{\varphi_{1,2}}.\end{array}\] More generally, we have $t_{i,j}=e^{i\varphi_{i,j}}\cos(\theta_{i,j}) \times q_{i,j} + r_{i,j}$ where $q_{i,j},r_{i,j}$ are terms depending uniquely on the the angles with lower indexes. We can notice there is at most one path from the $j^{th}$ input wire to the $i^{th}$ output wire involving $\theta_{i,j}$ and $\varphi_{i,j}$ and that $q_{i,j}\neq 0$ if and only if all $\theta_{k<i,j}$ and $\theta_{i,\ell<j}$ are nonzero. If one $\theta_{i,\ell<j}$ is zero, then we have $\varphi_{i,j}=\theta_{i,j}=0$ by the properties of the \T-circuits. If there are $K$ values of $\theta_{k<i,j}$ which are zero, then all the $K$ diagonals $\theta_{k,\ell'\geq j}$ are zero. By now considersing the path from the $(j+K)^{th}$ input wire to the $i^{th}$ output wire, we recover the same type of equation with $q_{i,j}\neq 0$. Now, we can substract $r_{i,j}$ and dividing by $q_{i,j}$, so that we have $e^{i\varphi_{i,j}}\cos(\theta_{i,j})=z_{i,j}$ with $z_{i,j}=(t_{i,j}-r_{i,j})/q_{i,j}$. If $z_{i,j}\neq0$ then $\theta_{i,j}\in [0,\frac{\pi}{2})$ and $\varphi \in [0,2\pi)$ are uniquely determined. If $z_{i,j}=0$, then $\theta_{i,j}=\frac{\pi}{2}$ and by the properties of \T-circuits, we have $\varphi_{i,j}=0$.

\begin{remark}The existence and the uniqueness have been shown in~\cite{clement2022lov} for very similar circuits that have two minor differences; the phases were on the top left of the beam splitters, and the range of the thetas and phases, expect the last layer, were all in $[0,\pi)$. We can therefore have an alternative proof by changing the strongly normalising and confluent rewriting system so the thetas are always in $[0,\frac{\pi}{2}]$ and the phases stay on the bottom left instead of the top left, without restricting their range.
\end{remark}

\section{Some properties of the \lovetitle-circuits}

\subsection{Choice of the generators}\label{app:notsobad}

The sources and detectors of the \love-calculus allow any and arbitrary finite support state on many modes, which may seem to be too powerful or far from the physical implementation. In that regard, we would like to highlight that:\begin{itemize} \item Some sources can directly generate more generic states such as a coherent superposition with the vacuum of the 2-photon state~\cite{loredo2019generation}, or even directly create entangled states~\cite{coste2023highrate}.
    \item Linear optical circuits are very modular, and each building block is usually used many times. It would therefore be more convenient to sometimes represent those building blocks directly by specifying what they do, instead of how they are implemented, as illustrated in \pref{fig:ex-Bell}.\item Optical interactions are very combinatorics, thus being unlikely to have a complete equational theory with only single mode sources\footnote{We can note that~\cite{felice2022path} bypasses that problem by allowing sums of diagrams}.
    \item This formalism still allows finding new results for linear optics, like the unique normal forms~\pref{sec:completeness}.\end{itemize}


\subsection{Remarkable properties}

In this subsection, we show a list of remarkable properties that either illustrate the behavior of some circuits or give more insight on how they work.

\textbf{Circuit with an identity wire have the semantics of a sum of smaller diagrams.}
\label{app:sumidwire}

\[\begin{array}{rcl}\interps{\tikzfigbox{0.8}{sumofdiagleft}}&=&\sum\limits_{k \in \K}\sum\limits_{\ell \in \setindex{L}}\langle \ell \mid k \rangle \interps{\tikzfigbox{0.8}{sumofdiagright1}} \\ &=& \sum\limits_{j \in \setindex{L} \cap \setindex{K}} \interps{\tikzfigbox{0.8}{sumofdiagright}} \end{array}\] %

\textbf{\interpspre{C} can always be expressed as a sum of $\Omega$.}

\label{app:sumOmega}

\begin{proposition}[Canonical decomposition]
    For every \love-circuit $C$, $\interpspre{C}$ can be expressed as a linear combination of $\Omegav{u}{v}(D)$, with $D$ a \lopp-circuit.
\end{proposition}

\begin{proof}
    This is a direct consequence of the linearity of $\interpspre{\cdot}$. With $\vect{f}=\sum\limits_{\vect{u}\in\setindex{U}}\alpha_{\vect{u}}\ket{\vect{u}}$, $\setindex{U}\finsubset \N^{\ntilde}$ and $\vect{g}=\sum\limits_{\vect{v}\in\setindex{V}}\beta_{\vect{v}}\bra{\vect{v}}$, $\setindex{V}\finsubset \N^{\mtilde}$, we have: 

\[    \begin{array}{rcl}
\interpspre{C}&=&\interpspre{\tikzfigbox{1.0}{Omegadec1}} \\
&=& \sum\limits_{\vect{u}\in \setindex{U}}\sum\limits_{\vect{v}\in \setindex{V}}\alpha_{\vect{u}}\beta_{\vect{v}}\interpspre{\tikzfigbox{1.0}{Omegadec2}} \\
&=& \sum\limits_{\vect{u},\vect{v}}\alpha_{\vect{u}}\beta_{\vect{v}}\Omegav{u}{v}(D) \\
&=&\sum\limits_{\vect{u},\vect{v}}\omega_{\vect{u},\vect{v}}\Omegav{u}{v}(D)
 
\end{array} \]

\end{proof}

 \section{Rewriting system of Section~\ref{sec:completeness}}

\subsection{\lovetitle~is complete for \lopp}
\label{app:loficompletelopp}
\begin{lemma}[\love~is complete for \lopp]\label{lem:derivLOppwithLOfi}
    We can derive the equations of \pref{fig:nLOpp} with the rules of the \love-calculus.
\end{lemma}

\begin{proof}
    All the rules are present aside from~(\ref{nLOpp:pM}). Note that~\pref{nLOpp:b0} can be derived with the rules of the PROP, and the Equations~(\ref{axiom:swap}) and~(\ref{axiom:E2}), as detailed in~\pref{app:proofcompLOpp}. We can now derive~(\ref{nLOpp:pM}) as follows:

    \begin{longtable}{rcl}
        \input{LOpp-D-left.tikz}&\eqtroiseqref{axiom:s0-0d}{axiom:p2pi}{nLOpp:b0}&\input{LOpp-D-left1.tikz}\\\\
        &\eqeqref{axiom:E2}&\input{LOpp-D-left2.tikz}\\\\
        &\eqtroiseqref{axiom:s0-0d}{axiom:p2pi}{nLOpp:b0}& 
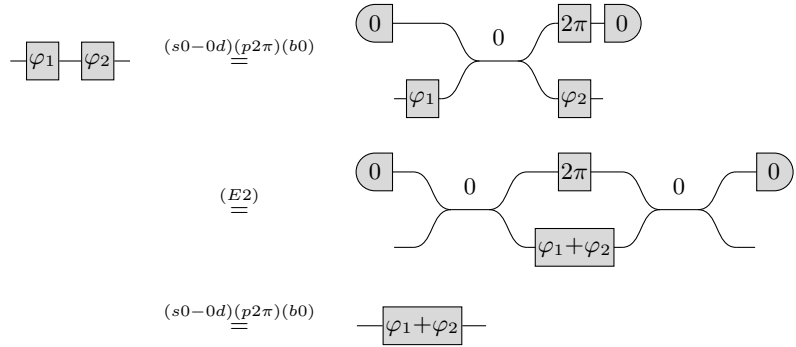
\begin{tikzpicture}
	\begin{pgfonlayer}{nodelayer}
		\node [style=none] (14) at (-1.725, 0) {};
		\node [style=phase] (15) at (0, 0) {$\varphi_1{+}\varphi_2$};
		\node [style=none] (16) at (1.675, 0) {};
	\end{pgfonlayer}
	\begin{pgfonlayer}{edgelayer}
		\draw (15) to (14.center);
		\draw (15) to (16.center);
	\end{pgfonlayer}
\end{tikzpicture}

    \end{longtable} 

\end{proof}

\subsection{Derivation of the rewriting system: Lemma~\ref{lem:derivrewrite}}
\label{app:derivrewrite}
Note that all equations including \lopp-circuits of~\pref{fig:rewritelofi} can be derived, as a direct collorary of \pref{lem:derivLOppwithLOfi}. For the other equations, we show that we can derive all equations of \pref{fig:usefuleq}, including the rules (\ref{rewritezerof}), (\ref{rewritezerog}) and (\ref{rewritewire}). The rule (\ref{rewriteremoveg}) is proven in \pref{lem:derivNm}.

\newcommand{\usefuleq}[5][1.0]{
    \begin{equation}  \tag{#3} \label{eq:#3} #2~\begin{array}{rcl} \scalebox{#1}{\input{#4.tikz}}\!\!\!\!\!&=&\!\!\!\!\!\scalebox{#1}{\input{#5.tikz}}\end{array}\end{equation}}
\begin{figure}
    \usefuleq{\qquad}{Rw}{Rwleft2}{Rwright2}
    \usefuleq{\qquad\qquad}{+F}{addortholeft}{addorthoFright}
    \usefuleq{\qquad\qquad}{+G}{addortholeft}{addorthoGright}
    \usefuleq{\qquad\qquad}{h1}{Aioneleft2}{Aioneright2}
    \usefuleq{\qquad\qquad\qquad}{s-md}{s-mdleft}{s-mdright}
    \usefuleq{\qquad\qquad\qquad}{sn-d}{sn-dleft}{sn-dright}

    \caption{\label{fig:usefuleq}Useful equations derived by the \love-calculus as show in~\pref{lem:usefulLOfieq}.~\pref{eq:Rw} is removing one identity wire when there are two, \pref{eq:+F} (resp.\ (\ref{eq:+G})) add any term in the source (resp.\ detector) trivially orthogonal on the last mode, \pref{eq:h1} is the same as \pref{axiom:h2} on one mode, Equations~(\ref{eq:s-md}) and (\ref{eq:sn-d}) are projections.}
\end{figure}

\begin{lemma}[Useful \love-equations]\label{lem:usefulLOfieq}
    We can derive the equations of \pref{fig:usefuleq} with the rules of the \love-calculus.
\end{lemma}

\begin{proof}
    We consider a linear function $h:\Fv{2} \rightarrow \Fv{2}$ such that for every $k\in\N$, $h(\ket{k,0})=\tilde{h}(\ket{k})$ and $(\bra{k,0})h=(\bra{k})\tilde{h}$.
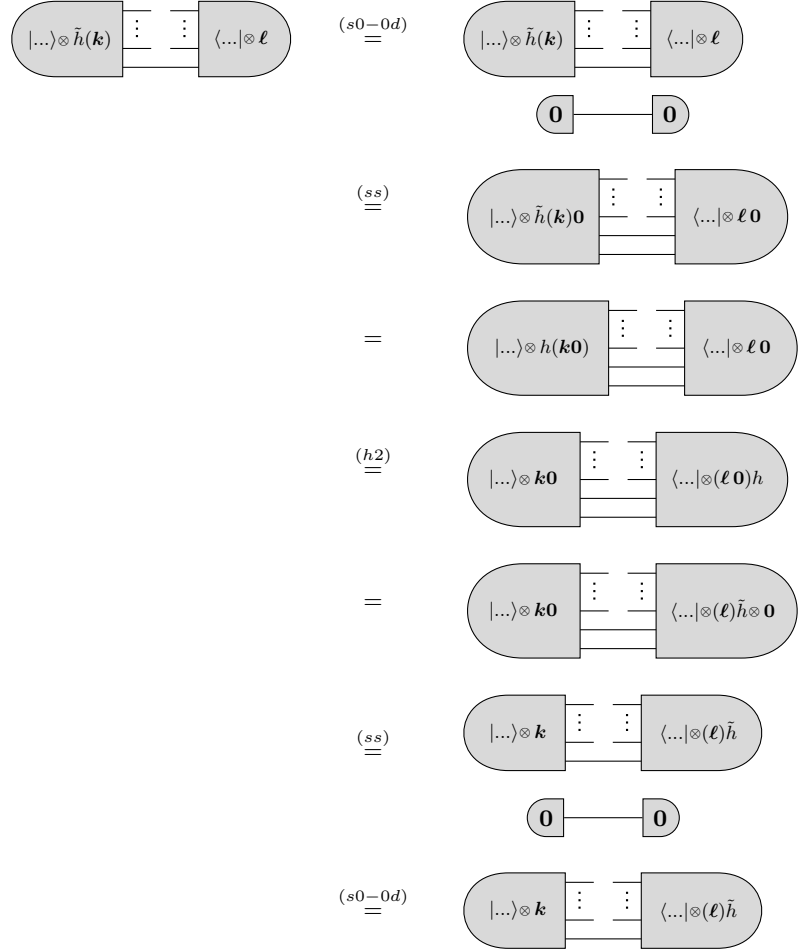
    \begin{longtable}{rcl}
        \input{Aioneleft2.tikz}&\eqeqref{axiom:s0-0d}&\input{Aioneleft2dec1.tikz}\\\\
        &\eqeqref{axiom:ss}&\input{Aioneleft2dec2.tikz}\\\\
        &=&\input{Aioneleft2dec3.tikz}\\\\
        &\eqeqref{axiom:h2}&\input{Aioneleft2dec4.tikz}\\\\
        &=&\input{Aioneleft2dec5.tikz}\\\\
        &\eqeqref{axiom:ss}&\input{Aioneleft2dec6.tikz}\\\\
        &\eqeqref{axiom:s0-0d}&\input{Aioneright2.tikz}\\\\

    \end{longtable}

    We consider a linear function $h:\Fv{2} \rightarrow \Fv{2}$ and a finite set $\setindex{K}\subset \N$ such that for every $k\in\setindex{K}$, $h(\ket{k,0})=\ket{k_1,k_2}$. 
    \begin{longtable}{rcl}
        \input{Rwleft2.tikz}&=&\input{Rwleft2dec1.tikz}\\\\
        &\eqeqref{axiom:h2}&\input{Rwleft2dec2.tikz}\\\\
        &\eqeqref{axiom:ss}&\input{Rwleft2dec3.tikz}\\\\
        &\eqeqref{axiom:s0-d}&\input{Rwright2.tikz}
    \end{longtable}

We consider a linear function $h:\Fv{2} \rightarrow \Fv{2}$ such that $h=h^{\dagger}$, $h$ is swapping the states $\ket{\tilde{\ell},0}$ and $\ket{\tilde{\ell},1}$ and is the identity everywhere else. We have:
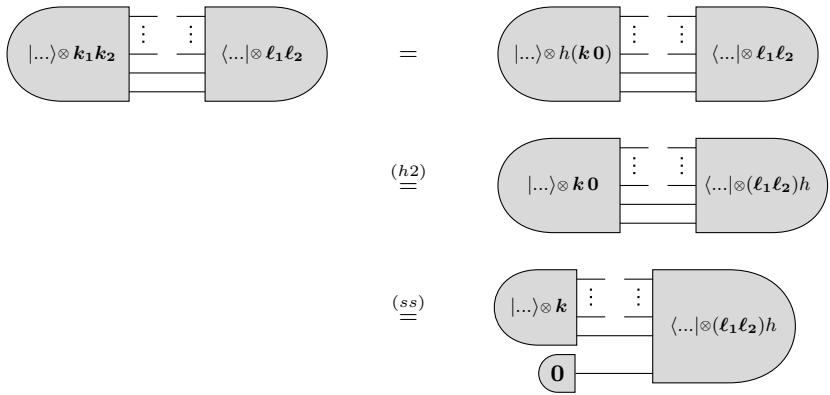
\begin{longtable}{rcl}
    \input{addorthoGright.tikz}&\eqdeuxeqref{axiom:s0-0d}{axiom:ss}&\input{addorthoright1.tikz}\\\\
    &=&\input{addorthoright2.tikz}\\\\
    &\eqeqref{axiom:h2}&\input{addorthoright3.tikz}\\\\
    &=&\input{addorthoright4.tikz}\\\\
    &\eqeqref{axiom:ss}&\input{addorthoright5.tikz}\\\\
    &\eqeqref{axiom:s0-d}&\input{addortholeft.tikz}
\end{longtable}

Similarly and symmetrically, we can derive \pref{eq:+F}.

We consider a linear map $\tilde{h}:\Fv{1}\rightarrow\Fv{1}$ such that $\tilde{h}(\ket{0})=\tilde{\vect{f}}$. We have:

\begin{longtable}{rcl}
    \input{s-mdleft.tikz}&\eqeqref{eq:Rw}&\input{s-mdleft1.tikz}\\\\
    &=&\input{s-mdleft2.tikz}\\\\
    &\eqeqref{eq:h1}&\input{s-mdleft3.tikz}\\\\
    &\eqeqref{axiom:s0-d}& \input{s-mdright.tikz} \\\\
\end{longtable}

Similarly and symmetrically, we can derive \pref{eq:sn-d}.

\end{proof}

\begin{lemma}\label{lem:derivNm}
    We can derive the equation~(\ref{rewriteremoveg}) in the \love-calculus:
    \begin{equation*}
        \begin{array}{rcl}
            \input{rewriteremoveg-sum-left.tikz}&=& \input{rewriteremoveg-sum-right.tikz}\\
        \end{array}
    \end{equation*}
\end{lemma}

\begin{proof}
    Let $\bra{\psi_L}=\sum_{\ell\in\setindex{L}}\xi_{\ell}\bra{\ell}$. We have $\vect{g}=\bra{\Niso_{\mtilde}(L)}\bra{\psi_{L}}+ \sum_{\ell\in\setindex{L}\setminus\{L\}}\bra{g_{\ell}}\bra{\ell}$. 
Let $\tilde{h}:\Fv{1}\rightarrow\Fv{1}$ be a linear function which is the identity on $\ket{i}$ for every $i\in(\setindex{K}\cup\setindex{L})\setminus \{L\}$, such that $\bra{L}\tilde{h}=\bra{\psi_L}$, and zero elsewhere. We can check that $\tilde{h}\ket{k}=\ket{k}+\xi_{k}\ket{L}$ for $k\neq L$, and $\tilde{h}\ket{L}=\xi_L\ket{L}$. We have:

\[\begin{array}{rcl}
    \vect{g}&=&\bra{\Niso_{\mtilde}(\ell)}\bra{\psi_{L}}+ \sum_{\ell\in\setindex{L}\setminus\{L\}}\bra{g_{\ell}}\bra{\ell} \\ 
    &=& \bra{\Niso_{\mtilde}(L)}\bra{L}\tilde{h} + \sum_{\ell\in\setindex{L}\setminus\{L\}}\bra{g_{\ell}}\bra{\ell}\tilde{h}\\ 
\end{array}\]

The linear function $\tilde{h}$ can therefore be removed with the equation~(\ref{eq:h1}), leading to:

\[\begin{array}{rcl}
    \vect{f}&=& \ket{f_L}\tilde{h}(\ket{L}) + \sum_{k\neq L} \ket{f_k}\tilde{h}(\ket{k}) \\ 
    &=&  \xi_L\ket{f_L}\ket{L}+\sum_{k\neq L}\ket{f_k}(\ket{k}+\xi_k\ket{L})\\ 
    &=& \left(\sum_{i\in\setindex{K}}\xi_i \ket{f_i}\right)\ket{L}+\sum_{k\neq L} \ket{f_k}\ket{k}
\end{array}\]

\end{proof}

\subsection{Strongly normalising: Proof of~Lemma~\ref{lem:strongly}}\label{app:strongly}

Given a circuit $C$ generated by the PRO version of the \love-circuits, we define $(x_1,x_2,x_3,\dots,x_{6})$ as:\begin{itemize}
    \item $x_1=\sum\limits_{i=0}^{n-2}(n-i)c(i)$ where $c(i)$ is the number of beam splitters in $C$ between the position $i$ and $i+1$.
    \item $x_2$ is the sum of the number of beam splitters with angles not in $[0,\frac{\pi}{2}]$, with those not in $[0,\pi)$ and those not in $[0,2\pi)$. For instance, a beam splitter with an angle of $3\pi$ will be counted three times, and those with $\frac{3}{2}\pi$ will be counted twice.
    \item To define $x_3$, let us define the depth of a phase shifter $p$ of $C$, denoted $d(p)$, as the maximal number of beam splitters that a photon starting from $p$ and going to the right would be able to traverse before reaching an output port, if it were allowed to choose each time whether to be reflected or
    transmitted. Then $e\coloneq \sum\limits_{\text{p phase shifter in }C}w(p)9^{d(p)}$ where, given a phase shifter $p=\begin{tikzpicture}
	\begin{pgfonlayer}{nodelayer}
		\node [style=phase] (0) at (0, 0) {$\varphi$};
	\end{pgfonlayer}
\end{tikzpicture}
$ of $C$, we define $w(p)$ as:
    
    $w(p)\coloneq\left\{
        \begin{array}{l}
        4\ \text{if $p$ belongs to a pattern of the form } \input{x5pattern.tikz} \\
        3 \text{if $p$ does not belong to such a pattern and if }\varphi\notin[0,2\pi) \\
        2 \text{in other case.}
        \end{array}
        \right.$
        \item $x_4$ is the number of identity wire connecting the sources and detectors.
        \item $x_5$ is the number of sources and detectors in $C$,
        \item to define $x_6$ let write in the Fock basis the generic terms in the sources as $\vect{f}=\sum\limits_{k_1,\dots,k_{\ntilde+1}}\alpha_{k_1,\dots,k_{\ntilde+1}}\ket{k_1,\dots,k_{\ntilde+1}}$ and in the detectors as $\vect{g}=\sum\limits_{\ell_1,\dots,\ell_{\mtilde+1}}\beta_{\ell_1,\dots,\ell_{\mtilde+1}}\ket{\ell_1,\dots,\ell_{\mtilde+1}}$. Then $x_6\coloneq\sum\limits_{\vect{f}\in\text{sources}}C_1(\vect{f})+\sum\limits_{\vect{g}\in\text{detectors}}\left(2C_2(\vect{g})-C_3(\vect{g})\right)$, with $C_1(\vect{f})\coloneq\#\{\alpha_{k_1,\dots,k_{\ntilde+1}}\neq 0\}, C_2(\vect{g})\coloneq\#\{\beta_{\ell_1,\dots,\ell_{\mtilde+1}}\neq 0\}$, and $C_3(\vect{g})\coloneq\#\{\beta_{\Niso_{\mtilde}(L),L}= 1, L\in \mathbb{N}\}$.
\end{itemize} 

We can highlight that $(x_1,x_2,x_3)$ are mainly inspired from~\cite{clement2022lov}. We conclude the proof by noticing that:
\begin{itemize}
    \item Rule~\ref{rewritephasemod2pi} strictly decreases $x_3$ without changing previous components of the tuple.
    \item Rule~\ref{rewritebsmod2pi} strictly decreases $x_2$ without changing $x_1$.
    \item Rule~\ref{rewritefusionphaseshifts} strictly decreases $x_3$ without changing previous components of the tuple.
    \item Rule~\ref{rewritezerophaseshifts} strictly decreases $x_3$ without changing previous components of the tuple.
    \item Rule~\ref{rewritezerobs} strictly decreases $x_1$.
    \item Rule~\ref{rewritetopphase} strictly decreases $x_3$ without changing previous components of the tuple.
    \item Rule~\ref{rewritepisur2} strictly decreases $x_3$ without changing previous components of the tuple.
    \item Rule~\ref{rewritethetabs} strictly decreases $x_2$ without changing $x_1$.
    \item Rule~\ref{rewriteminuspi} strictly decreases $x_2$ without changing $x_1$.
    \item Rule~\ref{rewriteE3} strictly decreases $x_1$.
    \item Rule~\ref{rewriteE2} strictly decreases $x_1$.
    \item Rule~\ref{rewritezerof} strictly decreases $C_1$ without changing anything else, therefore strictly decreasing $x_6$. 
    \item Rule~\ref{rewritezerog} strictly decreases $2C_2-C_3$, without changing anything else. Indeed, if we remove a coefficient where $\beta_{\Niso_{\mtilde}(L),L}=1$, then $-C_3$ increases by $1$ but $2C_2$ decreases by two. Therefore, $2C_2-C_3$ strictly decreases and $x_6$ is strictly decreasing.
    \item Rule~\ref{rewriteremoveg} can change the three values of $C_1,C_2$ and $C_3$, but still strictly decreases $x_6$. Let us consider the two cases: $(\xi_L\neq 1) \wedge (\forall \ell\neq L, \xi_{\ell}=0)$ and~$(\xi_L= 1) \wedge (\exists \ell\neq L, \xi_{\ell}=0)$. The first case doesn't change $C_1$ and $C_2$, but the term $-C_3$ strictly decreases by $1$. The second case doesn't change $C_3$, and the increase of $C_1$, i.e.\ the amount of new terms in $\vect{f}$, is bounded by $\#\{\xi_i\neq 0,i\neq L\}$, the number of terms removed in $\vect{g}$, which is the exact decrease of $C_2$. Therefore, $C_1+2C_2$ decreases by at least $\#\{\xi_i\neq 0,i\neq L\}>0$. We can check that the rule~\ref{rewritezerog} is not changing any other component of the tuple, and can conclude it only strictly decreases $x_6$.
    \item Rule~\ref{rewritewire} strictly decreases $x_4$ without changing previous components of the tuple.
    \item The rules coming from the axioms~\ref{axiom:ss} and~\ref{axiom:dd} strictly decrease $x_5$ without changing the previous components.
    \item The rules coming from the axioms~\ref{axiom:s-p} and~\ref{axiom:p-d} strictly decrease $x_3$.
    \item The rules coming from the axioms~\ref{axiom:s-b} and~\ref{axiom:b-d} strictly decrease $x_1$.
\end{itemize}

\subsection{Rewriting to the zero normal form: Proof of Remark~\ref{rem:NFzero}}\label{app:NFzeroproof}

In the normal form of \pref{def:NF}, it is possible that $\K=\emptyset$, meaning the semantics of the all circuit is the null function. Note the zero vector on $\Fv{2}$ can be written as $0=0 \otimes \ket{0}$ and the zero vector on $\Fvdual{2}$ as $0=0 \otimes \bra{1}$. We have:

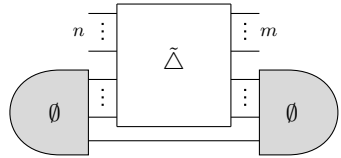
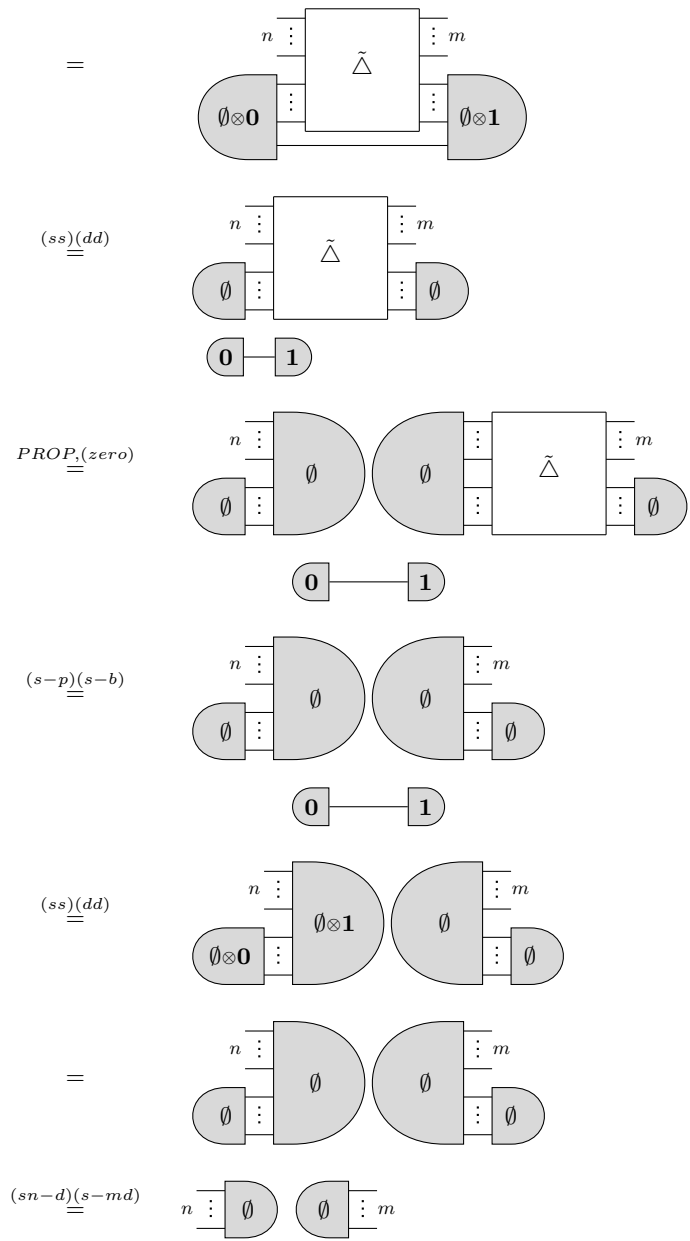
\begin{longtable}{rcl}
    &&\input{NFzero1.tikz}\\\\
    &=&\input{NFzero2.tikz}\\\\
    &\eqdeuxeqref{axiom:ss}{axiom:dd}&\input{NFzero3.tikz}\\\\
    &$\overset{PROP,(\ref{axiom:zero})}{=}$&\input{NFzero4.tikz}\\\\
    &\eqdeuxeqref{axiom:s-p}{axiom:s-b}&\input{NFzero5.tikz}\\\\
    &\eqdeuxeqref{axiom:ss}{axiom:dd}&\input{NFzero6.tikz}\\\\
    &=&\input{NFzero7.tikz}\\\\
    &\eqdeuxeqref{eq:sn-d}{eq:s-md}&   \input{NFzero8.tikz}\\\\
    \\\\

\end{longtable}

\section{Uniqueness of the normal forms of Section~\ref{sec:completeness}}
\label{app:uniqfullproof}
\subsection{Uniqueness of $T$ in the normal form: Lemma~\ref{lem:uniqT}}
\label{app:uniqT}

\begin{proposition}\label{prop:bspolynom}
    Let $\interpsone{\Ntriangle}=\interpsone{\input{bs-ps.tikz}}=\begin{pmatrix} e^{i\varphi}\costheta  & ie^{i\varphi}\sintheta \\ i\sintheta & \costheta \end{pmatrix}=\begin{pmatrix} ae^{i\varphi} & be^{i\varphi} \\ b & a \end{pmatrix}$. Let $k,\ell,n\in\N$, with $n>>k,\ell$ and $N=n+k-\ell$. For $a\neq 0$, we have:
    \[ \begin{array}{rcl}\bra{N,\ell}\interps{\input{bs-ps.tikz}}\ket{n,k}&=&e^{iN\varphi}a^N\sqrt{\frac{N!}{n!}} Q^{\ell}_k(n) \end{array}\] 
    where $Q^{\ell}_k$ is a real polynom of degree $\ell$.
\end{proposition}

\begin{proof} It directly follows from \pref{def:semmany}, by noting that $n-N=\ell-k$:
\[ \begin{array}{rcl}&&\bra{N,\ell}\interps{\input{bs-ps.tikz}}\ket{n,k}\\\\
    &=& e^{iN\varphi}\sqrt{\frac{N!\ell!}{n!k!}} \sum\limits_{\substack{p+q=N\\ \delta=p-q}}{\binom{n}{p}}{\binom{k}{q}}a^{k+\delta}b^{n-\delta} \ket{\ell_1,\ell_2}\\
    &=&e^{iN\varphi}\sqrt{\frac{N!\ell!}{n!k!}}\sum\limits_{q=0}^{k}{\binom{n}{N-q}}{\binom{k}{q}}a^{k+(N-q)-q}b^{n-((N-q)-q)}\\
    &=&e^{iN\varphi}\sqrt{\frac{N!\ell!}{n!k!}}\sum\limits_{q=0}^{k}{\binom{n}{n-N+q}}{\binom{k}{q}}a^{N+k-2q}b^{n-N+2q}\\
    &=&e^{iN\varphi}\sqrt{\frac{N!\ell!}{n!k!}}\sum\limits_{q=0}^{k}{\binom{n}{\ell-k+q}}{\binom{k}{q}}a^{N+k-2q}b^{(\ell-k)+2q}\\
     &=&e^{iN\varphi}a^N\sqrt{\frac{N!}{n!}}\left(\sqrt{\frac{\ell!}{k!}} \sum\limits_{q=0}^{k}\frac{n(n-1)\dots(n-\ell+k-q)}{(\ell-k+q)!}\binom{k}{q} a^{k-2q}b^{(\ell-k)+2q}\right)
     \\ &=&e^{iN\varphi}a^N\sqrt{\frac{N!}{n!}} Q^{\ell}_k(n)\end{array}\]

    \end{proof}

\paragraph*{Base case of \pref{lem:uniqT}}

\label{app:proofuniqT}
\begin{proof}
    Let $\interpsone{\Ntriangle}=\interpsone{\input{bs-ps.tikz}}=\begin{pmatrix} ae^{i\varphi} & be^{i\varphi} \\ b & a \end{pmatrix}$ and $\interpsone{\Ntriangle'}=\interpsone{\input{bs-ps-prime.tikz}}=\begin{pmatrix} a'e^{i\varphi'} & b'e^{i\varphi'} \\ b' & a' \end{pmatrix}$ with $0< a\leq 1$ and $0< a'\leq 1$ as\footnote{
        Note the case $a=0$ is not considered as $\Ntriangle$ or $\Ntriangle'$ would be a phase tensored the identity, and therefore wouldn't be a $\Tmn$-circuit.} $\theta,\theta'\in(0,\frac{\pi}{2}]$. Let $W=\sum\limits_{k,\ell}\omega_{k,\ell}\Omega^{k,\ell}(\Ntriangle)=\sum\limits_{k,\ell}\omega_{k,\ell}'\Omega^{k,\ell}(\Ntriangle')=W'$. Without loss of generality, let us consider that $0<a\leq a'\leq 1$. By hypothesis, there exists one $\omega_{k,\ell}'\neq 0$.  Let us consider the maximum $\ell_{max}$ such that there exists $k_m$ with $\omega_{k_m,\ell_{max}}\neq 0$. We have that for every $n>>\ell_{max},k$ with $N=k_m-\ell_{max}+n$:
\[ \begin{array}{crcl} 
    & \bra{N}W\ket{n} &=& \bra{N}W'\ket{n} \\
    \Leftrightarrow &e^{iN\varphi}a^N\sqrt{\frac{N!}{n!}} \left(\sum\limits_{k-\ell=k_m-\ell_{max}}\omega_{k,\ell} Q^{\ell}_k(n) \right) &=& e^{iN\varphi'}a'^N\sqrt{\frac{N!}{n!}} \left(\sum\limits_{k-\ell=k_m-\ell_{max}}\omega_{k,\ell}' Q'^{\ell}_k(n) \right) \\
 \end{array} \]
where $Q^{\ell}_k$ and $Q'^{\ell}_k$ are real polynomials of degree $\ell$ (\pref{prop:bspolynom}). Therefore for every such $n,N$ we have:

\[\begin{array}{rcl} 
    e^{iN(\varphi-\varphi')}\left(\frac{a}{a'}\right)^N\left(\sum\limits_{k-\ell=k_m-\ell_{max}}\omega_{k,\ell} Q^{\ell}_k(n) \right)=\left(\sum\limits_{k-\ell=k_m-\ell_{max}}\omega_{k,\ell}' Q'^{\ell}_k(n) \right)
\end{array} \]

As $\omega_{k_m,\ell_{max}}\neq 0$, the RHS is a nonzero complex polynomial of degree $\ell_{max}$ evaluated in $n$. Note the limit of the absolute value of the RHS is $\infty$ for $n\rightarrow\infty$. Then the LHS has necessarily a nonzero $\omega_{k,\ell}$ and $a'\leq a$, otherwise the limit of the absolute value of the LHS is $0$. Therefore, $a=a'$, thus $\cos(\theta)=\cos(\theta')$ and $\theta=\theta'$ as they are in $[0,\frac{\pi}{2}]$.

As $\theta=\theta'$, we have $a=a'$ and $b=b'$, thus $Q^{\ell}_k=Q'^{\ell}_k$ and for every $N>>k_m,\ell_{max}$ with $n=\ell_{max}-k_m+N$:

\[\begin{array}{rcl} 
    \sum\limits_{k-\ell=k_m-\ell_{max}}\left(e^{iN(\varphi-\varphi')}\omega_{k,\ell}-\omega_{k,\ell}'\right) Q^{\ell}_k(n) &=&0
\end{array} \]

Necessarily, there exists $N_0$ such that for every $N\geq N_0$, we have $e^{iN(\varphi-\varphi')}\omega_{k_m,\ell_{max}}=\omega_{k_m,\ell_{max}}'\neq 0$. Otherwise, the absolute value of the LHS doesn't converge to $0$ and therefore can't be equal to $0$ for every $N$. As it is true for every $N\geq N_0$, is it in particular true for $N=N_0$ and $N=N_0+1$, thus $e^{i(\varphi-\varphi')}=1$, thus $\varphi=\varphi'$.
\end{proof}

\paragraph*{Induction of \pref{lem:uniqT}}

For the induction step, we add arbitrary nonzero inputs on the first $k$ (resp.\ last $k'$) wires and arbitrary nonzero $\ell$ (resp.\ last $\ell'$) outputs on $\Ntriangleone$ and $\Ntriangletwo$. That projects both sides of the equation on a strictly smaller space. By the induction hypothesis, we can deduce an equality between smaller submatrices, i.e.\ $\submatrixg{\interpsone{\Ntriangleone}}{k}{k'}{\ell}{\ell'}=\submatrixg{\interpsone{\Ntriangletwo}}{k}{k'}{\ell}{\ell'}$. By taking complementary subspaces, we recover that $\submatrix{\interpsone{\Ntriangleone}}{m}{n}=\submatrix{\interpsone{\Ntriangletwo}}{m}{n}$. We conclude that $\Ntriangleone=\Ntriangletwo$ with \pref{prop:uniqSubT}.

\subsection{Linear independence of the $\Omega$-morphisms}

\begin{lemma}[Linear independence of $\Omega(\Ttilde)$]\label{lem:uniqOmega}
    For any $\Tmn$-circuit  $\Ntriangle:n+\ntilde\rightarrow m+\mtilde$, if
    $\sum_{\vect{u},\vect{v}}\omegav{u}{v}\Omegav{u}{v}(\Ntriangle)=\sum_{\vect{u},\vect{v}}\omegav{u}{v}'\Omegav{u}{v}(\Ntriangle)$ then  and $\omegav{u}{v}=\omegav{u}{v}'$ for every $\vect{u},\vect{v} \in \setindex{U},\setindex{V}$.
\end{lemma}
The proof of \pref{lem:uniqOmega} is the focus of Section~\ref{app:decompDelta} and~\ref{app:uniqOmega}.
In order to simplify the proofs, we show that the linear independence of $\Omega(\Ntriangle)$ with \Tmn-circuits is equivalent to a linear independence on \Tmnrec-circuits introduced in~\pref{def:Trec}. That is the following proposition. 

\begin{proposition}[Linear independence of $\Omega(\Trectilde)$]\label{prop:uniqOmegaDiamond}
    For any $\Tmn$-circuit  $\Ntriangle:n'+\ntilde\rightarrow m'+\mtilde$, there exist $n\leq n'$, $m\leq m'$ and a \Tmnrec-circuit $\Ndiamond:n+\ntilde\rightarrow m+\mtilde$ such that
    $\sum_{\vect{u},\vect{v}}\omegav{u}{v}\Omegav{u}{v}(\Ntriangle)=\sum_{\vect{u},\vect{v}}\omegav{u}{v}'\Omegav{u}{v}(\Ntriangle)$ if and only if $\sum_{\vect{u},\vect{v}}\omegav{u}{v}\Omegav{u}{v}(\Ndiamond)=\sum\limits_{\vect{u},\vect{v}}\omegav{u}{v}'\Omegav{u}{v}(\Ndiamond)$.
\end{proposition}

\begin{proof}
By rewriting the shape of \Ntriangle~with the rules of~\lopp, we can show that $\interpspre{\Ntriangle}=\interpspre{D'}\interpspre{(id^{\otimes n-\mtilde})\otimes\Ndiamond}\interpspre{D}$. By applying the inverse of those morphisms to the first equation with the $\Omega(\Ntriangle)$ and by reversibly removing the identities, we obtain the second one with the $\Omega(\Ndiamond)$. More details are provided in \pref{app:uniqOmegaDiamondproof}.
\end{proof}

\subsection{Equivalent decomposition with the creation operators}

 \label{app:decompDelta}
 It is standard in the second quantification formalism of quantum physics to look at the evolution of the operators acting on the system rather than directly the evolution of the state. In particular, we associate for each mode $i$ a creation operator \adag{i} such that $\adag{j}\ket{x_1,\dots,x_m}=\sqrt{x_j+1}\ket{x_1,\dots,x_j+1,\dots,x_m}$. This formalism has shown to be simpler to work with as given a linear optical circuit $D$, while the state evolutions are given by $\interpspre{D}$, the evolution of the creation operators are linear, hence \textit{linear optics}, and directly given by $\interpsone{D}$~\cite{kok2007review,heurtel2023slos}. We give here the examples of the evolution of \adag{j} through a phase shifter and a beam splitter, while the case for a generic circuit $D$ is described in \pref{prop:commutelambda}.

\begin{example}[\adag{j} through a phase shifter]\label{ex:adagps}
    Let $D=\begin{tikzpicture}[scale=0.85]
	\begin{pgfonlayer}{nodelayer}
		\node [style=none] (5) at (-1.1, 0) {};
		\node [style=none] (9) at (1.1, 0) {};
		\node [style=phase] (11) at (0, 0) {\small $\varphi$};
	\end{pgfonlayer}
	\begin{pgfonlayer}{edgelayer}
		\draw (5.center) to (11);
		\draw (11) to (9.center);
	\end{pgfonlayer}
\end{tikzpicture}
$ with $\interpsone{D}=e^{i\varphi}$. Through $\interpspre{D}$, the creation operator has the following transformation: $\adag{1} \mapsto e^{i\varphi}\adag{1}$. We can check it is sound with $\Pphipre\coloneqq\interpspre{\begin{tikzpicture}[scale=0.85]
	\begin{pgfonlayer}{nodelayer}
		\node [style=none] (5) at (-1.1, 0) {};
		\node [style=none] (9) at (1.1, 0) {};
		\node [style=phase] (11) at (0, 0) {\small $\varphi$};
	\end{pgfonlayer}
	\begin{pgfonlayer}{edgelayer}
		\draw (5.center) to (11);
		\draw (11) to (9.center);
	\end{pgfonlayer}
\end{tikzpicture}
}$. By first noticing $\ket{k}=\frac{(\adag{1})^{k}}{\sqrt{k!}}\ket{0}$, we have: 
    \[\begin{array}{ccccccccc} \Pphipre\ket{k}&=&\Pphipre\frac{(\adag{1})^{k}}{\sqrt{k!}}\ket{0}&=&\frac{(e ^{i\varphi}\adag{1})^{k}}{\sqrt{k!}}\Pphipre\ket{0}&=&  \frac{(e ^{i\varphi}\adag{1})^{k}}{\sqrt{k!}}\ket{0} &=& e^{ik\varphi}\ket{k} \\
    \end{array}\]
\end{example}

\begin{example}[\adag{j} through a beam splitter]\label{ex:adagbs}
    Let $D=\input{bs-xs.tikz}$ with $\interpsone{D}=\begin{pmatrix}\costheta & i\sintheta \\i\sintheta & \costheta \end{pmatrix}$. Similarly to \pref{ex:adagps} and as described in \pref{app:adagbs}, the creation operators have the following transformations: $\adag{1} \mapsto \costheta\adag{1}+i\sintheta\adag{2}$ and $\adag{2} \mapsto i\sintheta\adag{1}+\costheta\adag{2}$. 
\end{example}

\begin{proposition}[General commutation rule with $\interpsone{\cdot}$] \label{prop:commutelambda}
    For any \lopp-circuit $D:n+\ntilde\rightarrow m+\mtilde$ with $(d_{i,j})=\interpsone{D}$, we have: \[\Lambdav{u}\interpspre{D}=\prod\limits_{j=1}^{\ntilde} \left(\sum\limits_{i=1}^{n+\ntilde} d_{i,j}\opc{a}{i}\right)^{u_j}\interpspre{D}=\sum\limits_{\#\vect{s}+\#\vect{s'}=\#\vect{u}}\delta_{\vect{s},\vect{s'}}'\left(\Lambdav{s}\otimes\Lambdav{s'}\right)\interpspre{D}\]  where $\#\vect{x}\coloneqq\sum_k x_k$ and $\Lambdav{u}\coloneqq \prod_{j=1}^{\ntilde} (a_j^\dagger)^{u_j}$.

\end{proposition}
\begin{proof} We can prove $\interpspre{D}\adag{j}=(\sum_i d_{i,j}\adag{i})\interpspre{D}$ by commuting each $\adag{i}$ with all the phase shifters and beam splitters of $D$ with the formulas of Example~\ref{ex:adagps} and~\ref{ex:adagbs}. The first equality follows as the $\opc{a}{j}$ commute with each other. The second equality is shown by developing and noticing that $(\opc{a}{1})^{s_1}\dots (\opc{a}{\ntilde})^{s_{\ntilde}}(\opc{a}{\ntilde+1})^{s_1'}\dots (\opc{a}{\ntilde+n})^{s_n'}=\Lambdav{s}\otimes\Lambdav{s'}$.~\end{proof}

We can now introduce the following new decomposition.

\begin{lemma}[$\Delta$-decomposition of $\Omegav{u}{v}(\Ndiamond)$]\label{lem:Deltadec}
    For any $\Tmnrec$-circuit $\Ndiamond$, $\vect{u}\in \N^{\ntilde}, \vect{v}\in \N^{n}$: 
    \[\Omegav{u}{v}(\Ndiamond)=\sum\limits_{\substack{\vect{s}\preceq\vect{u},~\vect{t}\lesseqv\vect{v}\\\#\vect{s}-\#\vect{t}=\#\vect{u}-\#\vect{v}}}\delta_{\vect{s},\vect{t}}\Deltav{s}{t}(\Ndiamond)\] 
    where $\Deltav{s}{t}(\Ndiamond)\coloneqq (\Lambdav{s}\otimes id^{n}) \circ \Omegav{\vect{0}}{t}(\Ndiamond)$, $\delta_{\vect{u},\vect{v}}\neq 0$, $\prec$ is the lexicographical order, ie.\ $\vect{x}\prec\vect{u}$ if there exists $k$ such that $x_1=u_1,\dots,x_{k-1}=u_{k-1}$ and $x_k<u_k$, and \lessv~is the weak order on vectors such that $\vect{y}\lessv\vect{v}$ if there exists $\ell$ with $y_{\ell}<v_{\ell}$ and $y_{j}\leq v_{j}$ for $j\neq\ell$.
\end{lemma}

\begin{proof} Let us define $\prodsf{\vect{u}}\coloneqq\prod_{j=1}^{\ntilde} \sqrt{u_j!}$. First, we commute \interpspre{\Ndiamond} and $\Lambdav{u}$ as follows:
    \[\begin{array}{rcl}
        \Omegav{u}{v}(D) &  \overset{\text{Def.}~\ref{def:OmegaDelta}~(\Omega)}{=} & (id^{\otimes m} \otimes \bra{\vect{v}}) \interpspre{\Ndiamond} (id^{\otimes n} \otimes \ket{u})\\ 
        &\overset{\text{Def.}~\ref{prop:commutelambda}~(\Lambda)}{=}& (id^{\otimes m} \otimes \bra{\vect{v}}) \interpspre{\Ndiamond} (id^{\otimes n} \otimes \psqfac{u}^{-1}\Lambdav{u} \ket{\vect{0}}^{\ntilde}) \\ 
        &\overset{\text{Prop.}~\ref{prop:commutelambda}}{=}& \sum\limits_{\substack{(\vect{s},\vect{s'})\in(\setindex{S},\setindex{S'})\\\#\vect{s}+\#\vect{s'}=\#\vect{u}}}\psqfac{u}^{-1}\delta_{\vect{s},\vect{s'}}'\left(\Lambdav{s}\otimes\bra{\vect{v}}\Lambdav{s'}\right) \interpspre{\Ndiamond}(id^{\otimes n} \otimes\ket{\vect{0}}^{\ntilde})
    \end{array}\]

    \begin{itemize}
        \item If $s_j'>v_j$ then $\bra{\vect{v}}\Lambdav{s'}=0$. Otherwise, we have $\bra{\vect{v}}\Lambdav{s'}\propto\bra{\vect{v-s'}}$. By taking $\vect{t}=\vect{v}-\vect{s'}$, the sum is over the set $\setindex{T}=\left\{\vect{t}, 0\leq t_1 \leq v_1\wedge\dots\wedge 0\leq t_n \leq v_n \right\}=\left\{\vect{t}, \vect{t}\lesseqv \vect{v}\right\}$.
        \item The elements of $\setindex{S}$ are determined by the formula in \pref{prop:commutelambda}. As \Ndiamond~is a \Tmnrec-circuit, $d_{i,j}=0$ for $i<j-n$. We can show by developing the product of sums that necessarily $\deltav{s}{s'}'=0$ if $s\succ u$ and $\delta_{\vect{u},\vect{v}}'\neq 0$. That is illustrated in~\pref{app:Deltadecproof}. 
        \item $\#\vect{s}+\#\vect{s'}=\#\vect{u}$. As $\vect{t}=\vect{v}-\vect{s'}$, we have $\#\vect{t}=\#\vect{v}-\#\vect{s'}$ so $\#\vect{s}-\#\vect{t}=\#\vect{u}-\#\vect{v}$.
    \end{itemize}

By taking $\delta_{\vect{s},\vect{t}}=\psqfac{u}^{-1}\delta_{\vect{s},\vect{t}}'$, we conclude that $\Omegav{u}{v}(\Ndiamond)$ is equal to:
\[\begin{array}{clcc}
     &\sum\limits_{\substack{\vect{s}\preceq\vect{u},~\vect{t}\lesseqv\vect{v}\\\#\vect{s}-\#\vect{t}=\#\vect{u}-\#\vect{v}}}\delta_{\vect{s},\vect{t}}\left(\Lambdav{s}\otimes\bra{\vect{t}}\right) \interpspre{\Ndiamond}(id^{\otimes n} \otimes\ket{\vect{0}}^{\ntilde}) && \\ =& \sum\limits_{\substack{\vect{s}\preceq\vect{u},~\vect{t}\lesseqv\vect{v}\\\#\vect{s}-\#\vect{t}=\#\vect{u}-\#\vect{v}}}\delta_{\vect{s},\vect{t}}(\Lambdav{u}\otimes id^{n})\circ \Omegav{0}{t}(\Ndiamond)&=& \sum\limits_{\substack{\vect{s}\preceq\vect{u},~\vect{t}\lesseqv\vect{v}\\\#\vect{s}-\#\vect{t}=\#\vect{u}-\#\vect{v}}}\delta_{\vect{s},\vect{t}}\Deltav{s}{t}(\Ndiamond)
\end{array}\]

\end{proof}

\subsection{Uniqueness of the canonical decomposition} 
\label{app:uniqOmega}

Let us consider the scalar $z=\bra{\vect{y}}\Deltav{s}{t}(\Trec)\ket{\vect{x}}$. $\Lambda^{\vect{s}}$ adds $\vect{s_j}$ photons on the $j^{th}$ output wire, so we need to detect at least $y_j$ on the $j^{th}$ output mode, otherwise $z$ is zero. 
Furthermore, $\Omegav{\vect{0}}{\vect{t}}(\Trec)$ has no ancilla photon in the input, so to detect $t_{n}$ in the last detector, we need at least $t_{\ntilde}$ in the last input mode. If we have exactly $t_{n}$ in the last input mode, then we need at least $t_{n-1}$ photon in the second to last input mode, to be able to detect $t_{n-1}$ in the second last detector, otherwise $z=0$. We can therefore prove the following proposition.

\begin{proposition}[$\Deltav{s}{t}(\Ndiamond)$ threshold propeties]\label{prop:Deltathresh}
For any $\Tmnrec$-circuit $\Ndiamond:n+m\rightarrow n+m$ and $(\vect{s},\vect{t})\in(\N^n,\N^m)$, $\bra{\vect{y}}\Delta^{\vect{s},\vect{t}}(\Trec)\ket{\vect{x}}$ is nonzero for $(\vect{x},\vect{y})=(\vect{t},\vect{s})$ and is zero if\\ $(\vect{x}\lrev \vect{t})\lor (\vect{y}\lrev\vect{s})$, where $\lrev$ is the reverse lexicographical order, i.e.\ $\vect{y}\lrev\vect{v}$ if there exists $k$ such that $y_n=v_n,\dots,y_{k+1}=v_{k+1}$ and $y_k<v_k$.
\end{proposition}
\begin{proof}
    More generally, $\bra{\vect{y}}\Delta^{\vect{s},\vect{t}}(\Trec)\ket{\vect{x}}=0$ if there exists one $y_j$ such that $y_j<s_j$. We chose a stronger constraint with $\vect{y}\lrev\vect{s}$. The nonzero term for the case $(\vect{x},\vect{y})=(\vect{t},\vect{s})$ is a direct consequence of the definitions and the nonzero beam splitters of \Tmnrec-circuits.
\end{proof}

\begin{lemma}[Linear independence of $\Deltav{s}{t}(\Ndiamond)$]\label{lem:uniqDelta}
For any $\Tmnrec$-circuit $\Ndiamond:n+m\rightarrow n+m$, if
$\sum_{\vect{s},\vect{t}}\deltav{s}{t}\Deltav{s}{t}(\Ndiamond)=\sum_{\vect{s},\vect{t}}\deltav{s}{t}'\Deltav{s}{t}(\Ndiamond)$ then $\deltav{s}{t}=\deltav{s}{t}'$ for every $\vect{s},\vect{t} \in (\setindex{S},\setindex{T})$.
\end{lemma}
\begin{proof}
Let $\Delta=\sum_{\vect{s},\vect{t}}\deltav{s}{t}\Deltav{s}{t}(\Ndiamond)$ and $\Delta'=\sum_{\vect{s},\vect{t}}\deltav{s}{t}'\Deltav{s}{t}(\Ndiamond)$. By \pref{prop:Deltathresh}, we have $\bra{\vect{y}}\Deltav{x}{y}(\Ndiamond)\ket{\vect{x}}\neq 0$. Let us consider the equation $\bra{\vect{y}}\Delta\ket{\vect{x}}=\bra{\vect{y}}\Delta'\ket{\vect{x}}$. By starting with $\ket{\vect{x}}=\ket{0}^{\otimes n}$ and $\bra{\vect{y}}=\bra{0}^{\otimes m}$, we deduce with \pref{prop:Deltathresh} that $\deltav{0}{0}=\deltav{0}{0}'$. By instancing $\vect{y}$ with the successive values of $\setindex{S}$ ordered by $\lrev$, we deduce that for every $\vect{s}\in\setindex{S}$, $\deltav{s}{0}=\deltav{s}{0}'$. Then, for each $\vect{s}$ and by instancing $\vect{x}$ with the successive values of $\setindex{T}$ ordered by $\lrev$, we deduce that for every $\vect{s},\vect{t}$, $\deltav{s}{t}=\deltav{s}{t}'$.
\end{proof}

\paragraph*{Linear independence of $\Omega(\Ntriangle)$: Proof of \pref{lem:uniqOmega}}

    Let $W=\sum_{\vect{u},\vect{v}}\omegav{u}{v}\Omegav{u}{v}(\Ndiamond)\overset{\text{Lem.}~\ref{lem:Deltadec}}{=}\sum_{\vect{u},\vect{v}}\omegav{u}{v}\left(\sum_{\vect{s},\vect{t}}\deltav{s}{t}^{\vect{u},\vect{v}}\Deltav{s}{t}(\Ndiamond)\right)=\sum_{\vect{s},\vect{t}}\deltav{s}{t}\Deltav{s}{t}(\Ndiamond)$ and \\$W'=\sum_{\vect{u},\vect{v}}\omegav{u}{v}'\Omegav{u}{v}(\Ndiamond)\overset{\text{Lem.}~\ref{lem:Deltadec}}{=}\sum_{\vect{u},\vect{v}}\omegav{u}{v}'\left(\sum_{\vect{s},\vect{t}}\deltav{s}{t}^{\vect{u},\vect{v}}\Deltav{s}{t}(\Ndiamond)\right)=\sum_{\vect{s},\vect{t}}\deltav{s}{t}'\Deltav{s}{t}(\Ndiamond)$ with $W'=W$. We have $\deltav{u}{v}^{\vect{u},\vect{v}}\neq 0$ (\pref{lem:Deltadec}) and $\deltav{s}{t}=\deltav{s}{t}'$ (\pref{lem:uniqDelta}). By taking $\umaxi$ the maximum of $\setindex{U}$ with the total order $\prec$, and $\vmaxa$ a maximal element of $\{\vect{v}, \omega_{\umaxi,\vect{v}}\neq 0\}$ with the partial order $\lessv$, we have $\omega_{\umaxi,\vmaxa}\delta_{\umaxi,\vmaxa}^{\umaxi,\vmaxa}=\delta_{\umaxi,\vmaxa}=\omega_{\umaxi,\vmaxa}'\delta_{\umaxi,\vmaxa}^{\umaxi,\vmaxa}$, therefore $\omega_{\umaxi,\vmaxa}=\omega_{\umaxi,\vmaxa}'$. By repeating the same procedure with $W-\omega_{\umaxi,\vmaxa}\Omega^{\umaxi,\vmaxa}(\Ndiamond)$ and $W'-\omega_{\umaxi,\vmaxa}'\Omega^{\umaxi,\vmaxa}(\Ndiamond)$, we deduce that $\omegav{u}{v}=\omegav{u}{v}'$ for every $\vect{u},\vect{v} \in \setindex{U},\setindex{V}$, thus proving \pref{prop:uniqOmegaDiamond} and \pref{lem:uniqOmega}. \qed 

\section{Proofs of Appendix~\ref{app:uniqfullproof}}

\subsection{Proof of Proposition~\ref{prop:uniqOmegaDiamond}}
    \label{app:uniqOmegaDiamondproof}
    
    Let us consider a generic subtriangle $\Tmn:n+\tilde{n}\rightarrow m+\tilde{m}$.
    With the property \ref{propTt:noid}, we know there exists one beam splitter with a nonzero angle on the $(n+\tilde{n}-1,n+\tilde{n})$ wires. We note $\bigstar$ the one on the most right. As this beam splitter cannot be connected to the $\tilde{n}$ last input wires, there is necessarily one beam splitter on the $(n+\tilde{n}-1,n+\tilde{n})$ on its top left, and we also denote the most right of them with $\bigstar$. Up to the $(\tilde{n}-1,\tilde{n})$, we can show there is at least one nonzero beam splitter on the top left of the previous considered nonzero beam splitters, and we note with $\bigstar$ the most right of them, meaning all the other beam splitters on its right and before the have zero angles $0$. We denote $\anytheta$ (resp.\ $\anynztheta$) an arbitrary (resp.\ arbitrary but nonzero) beam splitter angle. With the property of \pref{def:triangular}, if there is a nonzero angle $\bigstar$ (resp.\ zero angle $0$), then we necessarily have nonzero (resp.\ zero) angles for all the beam splitters above (resp.\ under) in the right diagonal. The angles which are necessarily zero for the property 3 and 4 of \pref{def:subtriangular} are in red.
    
    Therefore, a $\Tmn$-circuit is always of the following form:  
    
    \[\tikzfigbox{0.65}{Tmnnonzero}\]

    Now, we can prove the following lemma.
    
    \begin{lemma}[\Tmnrec-circuit extraction] \label{lem:Trec}
        Any \Tmn-circuit can be rewritten to: 
        \[\tikzfigbox{1.0}{Trecgen2}\]
        where $D:n\rightarrow n$ and $D':m\rightarrow m$ are two \lopp-circuits and $\Ndiamond$ is a $\Tmnrecg{}{\ntilde}{}{\mtilde}$-circuit.
        \end{lemma}

        \begin{proof}
            With the completeness of~\lopp, we can change the shape of the circuit so the apex of the triangle is now at the bottom. More precisely, we can apply~\pref{oLOpp:E3} many times to move all the beam splitters under the diagonal of $\bigstar$ as follows:
    
    \[\tikzfigbox{0.65}{Tmnbluearrows}\]

    We conclude by noticing we have indeed the nonzero beam splitters as expected.

        \end{proof}
    
        We can now prove~\pref{prop:uniqOmegaDiamond}. Let $U=D\otimes id^{\otimes\ntilde}$, $U'=D'\otimes id^{\otimes\mtilde}$ and $I_k=id^{\otimes n-\mtilde}$. We have:
    
    \[\begin{array}{crcl}
        &\sum\limits_{\vect{u},\vect{v}}\omegav{u}{v}\Omegav{u}{v}(\Ntriangle)&=& \sum\limits_{\vect{u},\vect{v}}\omegav{u}{v}'\Omegav{u}{v}(\Ntriangle)  \\
        \Leftrightarrow&\sum\limits_{\vect{u},\vect{v}}\omegav{u}{v}\Omegav{u}{v}(U'\circ(I_k\otimes\Ndiamond)\circ U) &=& \sum\limits_{\vect{u},\vect{v}}\omegav{u}{v}'\Omegav{u}{v}(U'\circ(I_k\otimes\Ndiamond)\circ U) \\ 
        \Leftrightarrow&\sum\limits_{\vect{u},\vect{v}}\omegav{u}{v}\left(\interpspre{D'}\Omegav{u}{v}(I_k\otimes\Ndiamond)\interpspre{D}\right) &=& \sum\limits_{\vect{u},\vect{v}}\omegav{u}{v}'\left(\interpspre{D'}\Omegav{u}{v}(I_k\otimes\Ndiamond)\interpspre{D}\right)\\ 
        \Leftrightarrow&\interpspre{D'}\left(\sum\limits_{\vect{u},\vect{v}}\omegav{u}{v}\Omegav{u}{v}(I_k\otimes\Ndiamond)\right)\interpspre{D} &=& \interpspre{D'}\left(\sum\limits_{\vect{u},\vect{v}}\omegav{u}{v}'\Omegav{u}{v}(I_k\otimes\Ndiamond)\right) \interpspre{D}\\ 
        \Leftrightarrow&\sum\limits_{\vect{u},\vect{v}}\omegav{u}{v}\Omegav{u}{v}(I_k\otimes\Ndiamond)&=&\sum\limits_{\vect{u},\vect{v}}\omegav{u}{v}'\Omegav{u}{v}(I_k\otimes\Ndiamond)\\
        \Leftrightarrow&I_k\otimes\left(\sum\limits_{\vect{u},\vect{v}}\omegav{u}{v}\Omegav{u}{v}(\Ndiamond)\right)&=&I_k\otimes\left(\sum\limits_{\vect{u},\vect{v}}\omegav{u}{v}'\Omegav{u}{v}(\Ndiamond)\right)\\
        \Leftrightarrow&\sum\limits_{\vect{u},\vect{v}}\omegav{u}{v}\Omegav{u}{v}(\Ndiamond)&=&\sum\limits_{\vect{u},\vect{v}}\omegav{u}{v}'\Omegav{u}{v}(\Ndiamond)\\

    \end{array}\]
    
    \qed
    
\subsection{Example~\ref{ex:adagbs}: the creation operator through a beam splitter}
    \label{app:adagbs}
    Similarly to \pref{ex:adagps} with the phase shifter, let $D=\input{bs-xs.tikz}$ with $\interpsone{D}=\begin{pmatrix}\costheta & i\sintheta \\i\sintheta & \costheta \end{pmatrix}$. Through $\interps{D}$, the creation operators have the following transformations: $\adag{1} \mapsto \costheta\adag{1}+i\sintheta\adag{2}$ and $\adag{2} \mapsto i\sintheta\adag{1}+\costheta\adag{2}$. We can check it is sound with $\interps{\input{bs-xs.tikz}}$ defined in \pref{def:semmany}. Indeed, we have:
        \[\begin{array}{rcl} 
            && \interps{\input{bs-xs.tikz}}\ket{k_1,k_2} \\\\ 
            &=&\interps{\input{bs-xs.tikz}}\frac{(\adag{1})^{k_1}(\adag{2})^{k_2}}{\sqrt{k_1!k_2!}}\ket{0,0}\\ &=&\frac{1}{\sqrt{k_1!k_2!}}(\costheta\adag{1}+i\sintheta\adag{2})^{k_1}(i\sintheta\adag{1}+\costheta\adag{2})^{k_2}\interps{\input{bs-xs.tikz}}\ket{0,0} \\ 
            &=&  \frac{1}{\sqrt{k_1!k_2!}}(\costheta\adag{1}+i\sintheta\adag{2})^{k_1}(i\sintheta\adag{1}+\costheta\adag{2})^{k_2}\ket{0,0} \\ 
            &=& \frac{1}{\sqrt{k_1!k_2!}} \left(\sum\limits_{p=0}^{k_1}{\binom{k_1}{p}}(\costheta\adag{1})^{p}(i\sintheta\adag{2})^{k_1-p}\right)\left(\sum\limits_{q=0}^{k_2}{\binom{k_2}{q}}(i\sintheta\adag{1})^{q}(\costheta\adag{2})^{k_2-q}\right) \ket{0,0} \\
            &=&\frac{1}{\sqrt{k_1!k_2!}}\left(\sum\limits_{\ell_1=0}^{k_1+k_2}\sum\limits_{\substack{{p+q=\ell_1}\\{\delta=p-q}}}{\binom{k_1}{p}}{\binom{k_2}{q}}\costheta^{k_2+\delta}(i\sintheta)^{k_1-\delta}(\adag{1})^{\ell_1}(\adag{2})^{k_1+k_2-\ell_1}\right)\ket{0,0}\\
            &=&\sum\limits_{\ell_1+\ell_2=k_1+k_2}\sqrt{\frac{\ell_1!\ell_2!}{k_1!k_2!}} \sum\limits_{\substack{p+q=\ell_1\\ \delta=p-q}}{\binom{k_1}{p}}{\binom{k_2}{q}}\costheta^{k_2+\delta}(i\sintheta)^{k_1-\delta} \ket{\ell_1,\ell_2} \\
        \end{array}\]

\subsection{Proof of Lemma~\ref{lem:Deltadec}}\label{app:Deltadecproof}
We want to illustrate that $\textbf{s}\preceq \textbf{t}$ in the $\Delta$-decomposition of $\Omega^{u,v}(\Ndiamond)$. Let us consider the following example with $n=2$ and $\ntilde=3$: 

\[ \begin{array}{rcl} 
    \Omegav{u}{v}(\Ndiamond)&=&\interpspre{\tikzfigbox{0.8}{Trec32omega}} \\
    &=&\sum\limits_{\vect{s},\vect{t}}\deltav{s}{t}(\adag{1})^{s_1}(\adag{2})^{s_2}(\adag{3})^{s_3} \interpspre{\tikzfigbox{0.8}{Trec32emptyomega}} \\
    &=& \sum\limits_{\vect{s},\vect{t}}\deltav{s}{t}\Deltav{s}{t}(\Ndiamond) \end{array} \]  %

At first, we have $\interpspre{\Ndiamond}(\adag{3})^{u_1}(\adag{4})^{u_2}(\adag{5})^{u_3}$. The only path for the photons to be in the first output wire is by the third input wire. Therefore, we necessarily have $s_1\leq u_1$. If all the photons went into the first output wire, ie.\ $s_1=u_1$, then necessarily we have $s_2\leq u_2$, as the only paths that could add photons to the second output are the third and fourth input modes. Similarly, if $s_2=u_2$ then $s_3\leq u_3$. Therefore, we have $\vect{s}\preceq \vect{u}$.

\end{document}